\documentclass[11pt]{article}

\usepackage{amsthm, amsmath, amsfonts, amssymb, mathrsfs,bm}
\usepackage{latexsym, enumerate,color,graphicx}
\usepackage[margin=1in]{geometry}
\usepackage{algorithm,algorithmic}

\usepackage[backend=biber, style=numeric-comp, isbn=false, doi=false,eprint=false,url=false,natbib=true,maxbibnames=9,maxcitenames=2]{biblatex}

\usepackage{comment}

\RequirePackage[colorlinks,citecolor=blue,urlcolor=blue]{hyperref}

\renewcommand{\mathcal}{\mathscr}

\newcommand{\myx}{u}

\newcommand{\my}{y}
\newcommand{\mya}{a}
\newcommand{\myb}{b}
\newcommand{\myc}{c}

\newcommand{\mye}{e}

\newcommand{\mynu}{\nu}

\newcommand{\cX}{\mathcal{X}}

\newtheorem{lemma}{Lemma}
\newtheorem{prop}{Proposition}
\newtheorem{cor}{Corollary}
\newtheorem{rem}{Remark}

\newcommand{\sgn}{\operatorname{sgn}}

\begin{document}

\title{\vspace*{-3.3cm} On Parametric Optimal
  Execution and Machine Learning Surrogates  \footnote{We would like to thank Kevin Webster and Nicholas Westray, as well as an anonymous referee, for their very valuable comments and suggestions.} \footnote{The accompanying Jupyter Notebook is available at \url{https://github.com/moritz-voss/Parametric_Optimal_Execution_ML}.}  
  }

\author{Tao Chen
   \footnote{University of Michigan,
     Department of Mathematics, 530 Church Street Ann Arbor, MI 48109-1043, USA, email \texttt{chenta@umich.edu}.}
  \hspace{3ex}
  Mike Ludkovski
   \footnote{University of California Santa
     Barbara, Department of Statistics \& Applied Probability, Santa
     Barbara, CA 93106-3110, USA, email
     \texttt{ludkovski@pstat.ucsb.edu}.}
  \hspace{3ex}
  Moritz Vo{\ss}
   \footnote{University of California Los Angeles, Department of Mathematics, Los Angeles, CA 90095, USA, email \texttt{voss@math.ucla.edu}.}
}

\date{\today}

\maketitle
\begin{abstract}
We investigate optimal order execution problems in discrete time with instantaneous price impact and stochastic resilience. First, in the setting of linear transient price impact we derive a closed-form recursion for the optimal strategy, extending the deterministic results from~\citet{ObizhaevaWang:13}. Second, we develop a numerical algorithm based on dynamic programming and deep learning for the case of nonlinear transient price impact as proposed by~\citet{BouchaudEtAl:04}. Specifically, we utilize an actor-critic framework that constructs two neural-network (NN) surrogates for the value function and the feedback control. The flexible scalability of NN functional approximators enables parametric learning, i.e., incorporating several model or market parameters as part of the input space. Precise calibration of price impact, resilience, etc., is known to be extremely challenging and hence it is critical to understand sensitivity of the execution policy to these parameters. Our NN learner organically scales across multiple input dimensions and is shown to accurately approximate optimal strategies across a wide range of parameter configurations. We provide a fully reproducible Jupyter Notebook with our NN implementation, which is of independent pedagogical interest, demonstrating the ease of use of NN surrogates in (parametric) stochastic control problems.
\end{abstract}

\begin{description}
\item[Keywords:] optimal execution, parametric control, neural network surrogates, stochastic resilience
\end{description}

\section{Introduction}

In the past decade an extensive literature has analyzed optimal execution of trades within a micro-structural framework that accounts for the interaction between trades, prices, and limit order book liquidity. For example, one notable strand follows the approach of Obizhaeva and Wang \cite{ObizhaevaWang:13} which is able to generate closed-form formulas for optimal trading strategy with linear transient price impact. However, any practical use of these elegant mathematical derivations must immediately confront the strong dependence of the solution on the given model parameters. Concepts such as order book resilience or book depth are mathematical abstractions and are not directly available in the real world. Similarly, parameters such as inventory penalty, are model-specific and have to be entered by the user. Consequently, model calibration becomes highly nontrivial and in turn requires understanding the interaction between model parameters and the resulting strategy. In parallel, model risk, i.e., mis-specification of the dynamics, is also a major concern. Model risk can be partially mitigated by considering more realistic nonlinear models, but this comes at the cost of losing closed-form solutions. 

Motivated by these issues, in this article we approach optimal execution using the lens of \emph{parametric} stochastic control. To this end, we investigate numerical algorithms that determine the optimal execution strategy jointly in terms of the state variables (inventory and limit order book spread/mark-up), as well as model parameters (instantaneous price impact, resilience factor, inventory penalty, etc.). We consider augmenting 1-4 model parameters to the training space, yielding multi-dimensional control problems. 

Our contribution is two-fold. In terms of the numerical methods, we propose a direct approach to parametric control that focuses on generating functional approximators to the value function. A concrete choice of such \emph{statistical surrogates} that we present below are Neural Networks (NNs). Our approach employs NNs to approximate the Bellman equation and can be contrasted with other ways of using NNs, such as Deep Galerkin Methods \cite{AlAradi2018solving,AlAradi2019applications,Germain2021neural} for PDEs; or parameterization of the feedback control \cite{Pham18deep1,Pham18deep2} in tandem with stochastic gradient descent. One advantage of our method is its conceptual simplicity; as such it as an excellent \emph{pedagogical} testbed for machine learning methods. Indeed, given that optimal execution is now a core problem that is familiar to anyone working in mathematical finance, we believe that this example offers a great entryway to students or researchers who wish to understand and ``play" with modern numerical tools for stochastic control. To this end, we provide a detailed  Python Jupyter notebook that allows a fully reproducible checking of our results. Our notebook is made as simple as possible, in order to distill where the statistics tools come in, and aiming to remove much of the ``ML mystique" that can be sometimes present.
Our approach also emphasizes the questions of how to train the statistical surrogate, and how to approximate the optimal control, both important implementation aspects whose discussion is often skipped.

In terms of the financial application, we contribute to the optimal execution literature in several ways that would be of independent interest to experts in that domain. Specifically, our optimal execution modeling framework builds on the discrete-time, linear transient price impact model with exponential decay from~\citet{ObizhaevaWang:13}
and additionally allows for (i) nonlinear, power-type price impact \`a la \citet{BouchaudEtAl:04, BouchaudEtAl:09}, \citet{Gatheral:10}, \citet{AlfonsiFruthSchied:10}, \citet{Dang:17}, \citet{CuratoGatheralLillo:16}; (ii) a stochastic transient price impact driven by its own noise, akin to the continuous-time model  in~\citet{BechererBilarevFrentrup:18}; and (iii) a risk-aversion-type running quadratic penalty on the inventory as arising, e.g., in~\citet{SchiedSchoenebornTehranchi:10}. In other words, our model nests various existing proposals in the literature, offering a unified discrete-time framework that we investigate in detail. Other related work on optimal order execution with (exponentially) decaying (linear) transient price impact in discrete and continuous time include, e.g., \citet{AlfonsiFruthSchied:08}, \citet{AlfonsiSchied:10}, \citet{PredoiuShaikhetShreve:11}, \citet{AlfonsiSchiedSlynko:12}, \citet{GatheralSchiedSlynko:12}, \citet{LorenzSchied:13},
\citet{AlfonsiSchied:13},
\citet{AlfonsiAcevedo:14}, \citet{BankFruth:14},
\citet{AlfonsiBlanc:16}, \citet{FruthSchoenebornUrusov:14, FruthSchoenebornUrusov:19},  \citet{GraeweHorst:17}, \citet{LehalleNeuman:19}, \citet{HorstXia:19}, \citet{ChenHorstTran:19}, \citet{AckermannKruseUrusov:21_1, AckermannKruseUrusov:21_2}, \citet{FordeBetancourtSmith:21}, \citet{NeumanVoss:22}; we also refer to the recent monograph by~\citet{WebsterBook} for an excellent overview and discussion of price impact modeling. Moreover, in the linear case, we also establish a new explicit formula, see Proposition~\ref{prop:benchmark}, for the optimal execution strategy with stochastic transient price impact and inventory penalty, which extends the explicit deterministic solution from~\citet{ObizhaevaWang:13} and allows us to also accurately benchmark our machine learning approach. Therefore, our numerical experiments provide new reliable insights on the interaction between different model parameters and the optimal strategy. Obtaining these insights, in other words building a better intuition on how the model behaves in different regimes, was the original motivation for our work, and is valuable for practitioners who must develop gut feelings on how the model reacts as the real world (i.e.,~calibrated parameters) changes. In particular, our numerical analysis complements the studies on deterministic optimal execution strategies with non-linear transient price impact carried out in~\citet{Dang:17} and~\citet{CuratoGatheralLillo:16}.

In the broader context of machine learning methods for stochastic optimal control, our work is related to other applications of neural networks to financial problems, see \citet{Pham18deep1,Pham18deep2,IsmailPham:19r}. Perhaps the closest is~\citet{LealLauriereLehalle:21} who also study execution problems, but in a model with only temporary and permanent price impact \`a la~\citet{BertsimasLo:98}, \citet{AlmgrenChriss:01}, \citet{CarteaJaimungal:16}; see also \citet{Papanicolaou23} for a similar study. For other approaches in optimal execution in the presence of temporary price impact more in the flavor of reinforcement learning see, e.g.,  the recent survey articles by \citet{HamblyXuYang:21} and \citet{Jaimungal:22} and the references therein.

This article is organized as follows. Section \ref{sec:model} formulates our generalized optimal execution setting. Section \ref{sec:benchmarksol} presents an explicit reference solution for unconstrained trading strategies. Section \ref{sec:method} describes our methodology for parametric stochastic control via statistical surrogates. Section \ref{sec:results} presents the numerical experiments and resulting insights. Following a brief conclusion in Section \ref{sec:conclude}, Section \ref{sec:proofs} contains the proofs.

\section{Problem formulation}\label{sec:model}

Let $(\Omega,\mathcal{F}, \mathbb{F} = (\mathcal{F}_n)_{n=0,\ldots,N},\mathbb{P})$ be a discrete-time filtered probability space with trivial $\sigma$-field $\mathcal{F}_0$ and a terminal time horizon $N \in \mathbb{N}$. We consider a financial market with one risky asset whose $\mathbb{F}$-adapted real-valued unaffected fundamental price process is denoted by $P=(P_n)_{n=0,\ldots,N}$. We set $P_0 := p_0 \in \mathbb{R}_+$.

Suppose a large trader dynamically trades in the risky security and incurs price impact in an adverse manner. Specifically, for $n\in\{1, \ldots, N\}$, by choosing her number of shares in the risky asset at time $n-1$, she trades $\myx_n  \in \mathcal{F}_{n-1}$ shares in the $n$-th trading period and confronts the $n$-th fundamental random shock $\Delta P_n := P_n - P_{n-1} \in \mathcal{F}_n$. Her action permanently affects the future evolution of the mid-price process $M = (M_n)_{n=0,\ldots,N}$ which becomes
\begin{equation} \label{def:midprice}
  M_n := P_n + \gamma
  \sum_{j=1}^n \myx_j, \quad n=1,\ldots,N,
\end{equation}
after the $n$-th order $\myx_n$ is executed (we set $M_0 := P_0$). The parameter $\gamma \geq 0$ represents linear permanent price impact. 

In addition, the trader's market orders are filled at a deviation $D_n$ from the mid-price $M_n$ in~\eqref{def:midprice}. The \emph{post}-execution dynamics of these deviations $D=(D_n)_{n=1,\ldots,N}$ from the mid-price after trading $\myx_n$ shares at time $n \in \{1,\ldots,N\}$ are modeled as
\begin{equation}\label{def:deviation}
  \begin{aligned}
    D_0 & := d_0, \\
    D_n & := (1-\kappa) D_{n-1} + \eta  \vert \myx_n \vert^{\alpha} \sgn(\myx_n)  +
    \epsilon_n,
    \quad n = 1, \ldots, N,
  \end{aligned}
\end{equation}
with $d_0 \in \mathbb{R}$ denoting the given initial deviation. Thus, in the absence of the large trader's actions, the deviation tends to revert exponentially to zero at rate $\kappa \in (0,1]$, with the latter parameter known as the book \emph{resilience}. The resilience captures the transience of the instantaneous price impact, with $\kappa$ representing the fraction by which the deviation from the mid-price $M$ incurred by past trades diminishes over a trading period. Empirically, the deviation process~$D$ can be calibrated by considering the limit order book (LOB) spread and depth. 

Due to finite market depth, which is measured by $1/\eta > 0$, the trader's turnover of $\myx_n$ shares pushes the deviation in the trade's direction by a constant factor $\eta$ times the instantaneous price impact, which is assumed to be $|\myx_n|^\alpha \sgn(\myx_n)$, $\alpha > 0$. Following, e.g., \citet{BouchaudEtAl:09, BouchaudEtAl:04}, the latter power-type term generalizes the common linear situation $\alpha=1$ from~\citet{ObizhaevaWang:13} where the instantaneous price impact is $\eta \myx_n$. However, empirically the instantaneous price impact is observed to be concave (at least for relatively small $\myx_n$), so that $\alpha < 1$ seems more realistic; see~\citet{LilloFarmerMantegna:03}, \citet{BouchaudEtAl:09}, \citet{BacryIugaLasnier:15}. 

The $\mathbb{R}$-valued, $\mathbb{F}$-adapted sequence of zero-mean random variables $(\epsilon_n)_{n=1,\ldots,N}$ represents additional small perturbations in the deviation stemming, e.g., from market and limit orders, which are placed at time $n$ by other small market participants, making the transient price impact captured by the deviation process $D$ stochastic; cf., e.g.,~\citet{BechererBilarevFrentrup:18}.  

For the rest of the paper, we assume that $(\Delta P_n)_{n=1,\ldots,N}$ are independent and square integrable  random variables with mean zero. The perturbations $(\epsilon_n)_{n=1,\ldots,N}$ in~\eqref{def:deviation} are i.i.d.~normally distributed random variables with mean zero and variance $\sigma^2 >0$, and independent of $(\Delta P_n)_{n=1,\ldots,N}$ as well. We also refer to the case $\sigma = 0$ where $\epsilon_n \equiv 0$ for every $n \in \{1,\ldots,N\}$. Finally, we let the filtration~$\mathbb{F}$ be given by $\mathcal{F}_n = \sigma(\{\Delta P_1, \epsilon_1,\ldots,\Delta P_n,\epsilon_n\})$ for all $n \in \{1, \ldots, N\}$. 

\subsection{Optimal Trade Execution} \label{subsec:execution}

From now on, we suppose that the large trader wants to carry out a \emph{buying program} to buy $X_0>0$ shares by executing $N$ market \emph{buy orders} $u_n \geq 0$, $n=1,\ldots,N$. She starts with zero inventory and we use $X_n$ to denote the remaining number of shares she needs to buy after step $n$ to reach her target. That is, we set 
\begin{equation} \label{def:inventory}
X_{n} := X_0 - \sum_{j=1}^{n} \myx_j, \qquad n=1,\ldots,N,   
\end{equation}
where $X_n$ represents her remaining order to be filled \emph{after} she executed her $n$-th trade $\myx_{n}$. 

As common in the literature, in order to describe the evolution of the large trader's cash balance, we assume that the $n$-th transaction $\myx_n$ affects the mid-price $M$ and deviation $D$ gradually. More precisely, half of the $n$-th order is filled at the \emph{pre-}transaction's mid-price $M_{n-1}$, as well as the refreshed \emph{pre-}transaction's deviation $(1-\kappa) D_{n-1}$, whereas the other half is executed at the less favorable \emph{post-}transaction quantities $M_{n} - \Delta P_n$ (before the $n$-th fundamental random shock hits the stock price) and $D_n-\epsilon_n$ (before the $n$-th random perturbation of the deviation). Hence, assuming zero interest rates, the self-financing condition dictates that changes in the trader's cash balance $(C_n)_{n=1,\ldots,N}$ with initial value $c_0 \in \mathbb{R}$ are only due to her buying activity of the risky asset which is executed at the previously
described \emph{average} execution prices:
\begin{equation} \label{def:cashdynamics}
  \begin{aligned}
    C_0 := & \; c_0, \\
    C_n := & \; C_{n-1} - \left( \frac{M_{n-1} + (1-\kappa) D_{n-1} + M_{n} - \Delta P_n + D_n -\epsilon_n}{2} \right) \myx_n \\
    = & \; C_{n-1} - \left( P_{n-1} - \frac{\gamma}{2} (X_{n}
      + X_{n-1}) + \gamma X_0 \right) \myx_n\\
    & - \left( (1-\kappa) D_{n-1} + \frac{\eta}{2} \myx_n^\alpha 
    \right)  \myx_n, \qquad n = 1, \ldots, N.
  \end{aligned}
\end{equation}

\begin{lemma}
The terminal cash position $C_N$ at time $N$ of a buying schedule $(u_n)_{n=1,\ldots,N}$ with $u_n \geq 0$ for all $n\in\{1,\ldots,N\}$ and terminal state constraint $X_N = 0$ is given by
\begin{equation}\label{eq:lemma1}
  \begin{aligned}
    C_N =& \; c_{0} - \left( p_0 + \frac{\gamma}{2} X_0 \right) X_0-
    \sum_{n=1}^{N-1} X_n \Delta P_n - \sum_{n=1}^N \left( (1-\kappa) D_{n-1} + \frac{\eta}{2} 
      u_n^{\alpha} \right) u_{n} .
  \end{aligned}
\end{equation}
\end{lemma}

\begin{proof}
Using $\myx_n = X_{n-1}-X_n$ we obtain from~\eqref{def:cashdynamics} together with~\eqref{def:deviation} and~\eqref{def:midprice}
\begin{align}
  C_N =
  & \; C_{0} - \sum_{n=1}^{N} P_{n-1} \myx_{n} +
    \frac{\gamma}{2} \sum_{n=1}^N (X_{n-1}^2 - X_{n}^2) - \gamma X_0
    \sum_{n=1}^N \myx_{n} \nonumber \\
  & - \sum_{n=1}^N \left( (1-\kappa) D_{n-1} + \frac{\eta}{2} \myx_n^{\alpha} \right) \myx_n \nonumber \\
  = & \; c_{0} - \sum_{n=1}^{N-1} X_n \Delta P_n - p_0 X_0 - \gamma X_0 (X_0 - X_N) +
      \frac{\gamma}{2} (X_0^2 - X_N^2) 
      \nonumber \\
  & - \sum_{n=1}^N \left( (1-\kappa) D_{n-1} + \frac{\eta}{2} \myx_n^{\alpha} \right)
    \myx_n.\label{eq:terminalcash}
\end{align}
Under the terminal state constraint $X_N = 0$,
the representation in~\eqref{eq:terminalcash} simplifies to \eqref{eq:lemma1}
\end{proof}

Next, in order to introduce the trader's optimization problem let us denote for all time steps $n=1,\ldots,N$ the collection of admissible buying strategies by 
\begin{equation} \label{def:admissibleSet}
    \mathcal{A}_n := \left\{ (\myx_j)_{j=n,\ldots,N} : \, \myx_j \in L^{(\alpha+1) \vee 2}(\mathcal{F}_{j-1}, \mathbb{P}),\, \myx_j \geq 0\;\text{a.s.} \; \text{for all} \; j=n,\ldots, N, \; \sum_{j=n}^N \myx_j = X_{n-1} \right\}.
\end{equation}
The trader aims to maximize her expected terminal cash position given in~\eqref{eq:lemma1} while also controlling for inventory risk. The latter is modeled through a quadratic urgency penalty $\mynu X_n^2$ for an urgency parameter $\mynu \ge 0$ on her outstanding order. Combining the two terms, the objective is to minimize
\begin{equation} \label{def:optProblem}
  \inf_{(\myx_n)_{n=1,\ldots,N}\in\mathcal{A}_1}
  \mathbb{E}\left[ \sum_{n=1}^N \left\{\left( (1-\kappa) D_{n-1} +\frac{\eta}{2} \myx_{n}^{\alpha} \right) 
    \myx_{n} + \mynu (X_{n-1} - \myx_n)^2 \right\}\right].
\end{equation} 
Note that $X_{n-1}-\myx_n = X_n$ is the remaining order to be filled; our notation emphasizes the role of $X_{n-1}$ and $D_{n-1}$ as the state variables, and $\myx_n$ as the control. As it is well-known in the literature, we remark that the permanent impact $\gamma$ disappears in \eqref{def:optProblem} and from our further discussion since it only adds a fixed offset $-0.5\gamma X_0^2$ to the terminal cash position, irrespective of the trading strategy. Similarly, the martingale term $\sum_n X_n \Delta P_n$ in \eqref{eq:lemma1} disappears as well after taking expectations thanks to the independence of price increments.

We introduce for all $n \in \{1,\ldots,N\}$ the value function as
\begin{equation}\label{def:valuefunction}
    \begin{aligned}
    & V_n(x, d) := \inf_{(\myx_j)_{j=n,\ldots,N} \in \mathcal{A}_{n} } \mathbb{E}\Bigg[ \sum_{j=n}^N \bigg\{ \left( (1-\kappa) D_{j-1} +\frac{\eta}{2} \myx_{j}^{\alpha}\right) \bigg.
    \myx_{j} \Bigg. \\
    & \Bigg. \hspace{195pt} \bigg. + \mynu (X_{j-1} - \myx_j)^2 \bigg\} \, \bigg\vert \, X_{n-1}=x, D_{n-1}=d \Bigg].
    \end{aligned}
\end{equation}

To characterize $V_n$, we use the corresponding dynamic programming (DP) equation.
Since at the last period $N$, the admissible set $\mathcal{A}_N$ is a singleton $\myx_N \equiv X_{N-1}$, we have the terminal condition
\begin{align} \label{eq:bellman-1}
V_{N}(x, d) 
    & = (1-\kappa) \cdot d \cdot x + \frac{\eta}{2} x^{\alpha+1},
\end{align}
and then for $n=N-1,\ldots,1$
\begin{align}\label{eq:bellman-2}
V_{n}(x, d) 
    & = \inf_{\myx \in [0,x]} \mathbb{E}\Bigg[ (1-\kappa) \cdot d \cdot \myx + \frac{\eta}{2} \myx^{\alpha+1}  + \mynu ( x - \myx)^2  \Bigg. \nonumber \\
    & \hspace{52pt} \Bigg. + V_{n+1}(x-\myx,D_{n}) \, \bigg\vert \, X_{n-1}=x, D_{n-1}=d \Bigg]
\end{align}
with  $D_n = (1-\kappa)d + \eta u^\alpha + \epsilon_n$ as postulated in~\eqref{def:deviation} and the expectation being with respect to the Gaussian noise $\epsilon_n \sim \mathcal{N}(0,\sigma^2)$. 

The DP equation \eqref{eq:bellman-2} has two primary state variables $(X_{n-1}, D_{n-1})$. Below, we will also consider its dependence on the static parameters $\kappa, \eta, \alpha, \mynu, \sigma$. Recall that $\kappa \in (0,1]$ is the book resilience (smaller $\kappa$ increases the transient price impact); $\eta > 0$ is the instantaneous price impact (larger $\eta$ makes trades affect $D_n$ more); $\mynu \ge 0$ is the urgency parameter (larger $\mynu$ encourages larger buys to mitigate inventory risk); $\alpha \approx 1$ is the exponent of the instantaneous price impact function ($\alpha \gtrless 1$ leads to convex (resp.~concave) price impact) and $\sigma >0$ is the standard deviation of the one-step-ahead deviation $D_n$. Note that while $\kappa$ and $\alpha$ are dimensionless, the value of $\mynu$ should be thought of relative to the initial inventory $X_0$, and the values of $\eta,\sigma$ should be picked relative to fluctuations in $D_n$. For example, if $X_0 = 10^5$ (buying program of a hundred thousand shares) then $\mynu$ should be on the order of $10^{-4}$, $\eta$ should be on the order of $10^{-3}$ (so that $D_n$ is on the order of 10-100), and $\sigma$ should be on the order of 1.

\begin{rem}[Unconstrained Problem] In the above formulated optimal trade execution problem it is tempting to a priori allow for trading in both directions (buy orders $u_n > 0$ and sell orders $u_n < 0$) as in the discrete-time linear transient price impact model in~\citet{AckermannKruseUrusov:21_1}; i.e.,
to compute
\begin{equation}\label{def:valuefunctionUC}
    \begin{aligned}
    & V^\circ_n(x, d) := \inf_{(\myx_j)_{j=n,\ldots,N} \in \mathcal{A}_{n}^\circ } \mathbb{E}\Bigg[ \sum_{j=n}^N \bigg\{ \left( (1-\kappa) D_{j-1} +\frac{\eta}{2} \vert \myx_{j} \vert^{\alpha} \sgn(u_j) \right) \bigg.
    \myx_{j} \Bigg. \\
    & \Bigg. \hspace{195pt} \bigg. + \mynu (X_{j-1} - \myx_j)^2 \bigg\} \, \bigg\vert \, X_{n-1}=x, D_{n-1}=d \Bigg]
    \end{aligned}
\end{equation}
over the set of unconstrained order schedules
\begin{equation} \label{def:admissibleSetUC}
    \mathcal{A}_n^\circ := \left\{ (\myx_j)_{j=n,\ldots,N} : \, \myx_j \in L^{(\alpha+1) \vee 2}(\mathcal{F}_{j-1}, \mathbb{P}),\; \mathbb{R} \text{-valued for all} \; j=n,\ldots, N, \; \sum_{j=n}^N \myx_j = X_{n-1} \right\}.
\end{equation}
However, as it will become apparent in Section~\ref{sec:benchmarksol}, the exogenous noise $(\epsilon_n)_{n=1,\dots,N}$ in the deviation process $D$ in~\eqref{def:deviation} would then trigger price manipulation in the sense of~\citet{HubermanStanzl:04}. That is, there would exist profitable round-trip trades, i.e., nonzero strategies $(u_n)_{n=1,\ldots,N} \in \mathcal{A}^\circ_1$ that generate strictly negative expected costs with $X_0 = 0 = X_N$  by exploiting a nonzero deviation $D_n$ for some $n \in \{0,1,\ldots,N\}$. This can be ruled out by either introducing a bid-ask spread or confining trading in one direction only; see also the discussion on price manipulation in~\cite{AckermannKruseUrusov:21_1, FruthSchoenebornUrusov:19, FruthSchoenebornUrusov:14} for the linear case $\alpha=1$, as well as \cite{AlfonsiSchied:10,Gatheral:10, CuratoGatheralLillo:16} for the nonlinear case $\alpha \neq 1$.
\end{rem}

\section{Explicit Solution for the Unconstrained Linear Case} \label{sec:benchmarksol}

The unconstrained optimal trade execution problem in~\eqref{def:valuefunctionUC} and~\eqref{def:admissibleSetUC} can be solved explicitly in the case of linear transient price impact ($\alpha =1$) following a similar computation as done by~\citet{ObizhaevaWang:13}. They derived a solution in the deterministic case ($\sigma=0$) without inventory penalty ($\nu=0$) and with initial deviation $d_0 = 0$.

\begin{prop} \label{prop:benchmark}
Let $\alpha=1$ and let $\sigma \geq 0$. Define
\begin{equation} \label{prop:benchmark:recursionTerminal}
\mya_N := \frac{\eta}{2}, \; \myb_N := 1 - \kappa, \; \myc_N := 0
\end{equation}
and, recursively, for all $n=N-1,\ldots,1$, set
\begin{equation} \label{prop:benchmark:recursion}
\left\{\begin{aligned}
    \mya_n := & \, \mynu + \mya_{n+1} - \frac{\left( 2 \mynu + 2 \mya_{n+1} - \eta \myb_{n+1} \right)^2}{2 \eta + 4 \mynu + 4 \mya_{n+1} - 4 \eta \myb_{n+1} + 4 \eta^2 \myc_{n+1}} , \\
    \myb_n := & \, (1-\kappa) \myb_{n+1} + \frac{2 (1-\kappa) \left( 2 \mynu + 2 \mya_{n+1} - \eta \myb_{n+1} \right) (1 - \myb_{n+1} + 2 \eta \myc_{n+1} ) }{2 \eta + 4 \mynu + 4 \mya_{n+1} - 4 \eta \myb_{n+1} + 4 \eta^2 \myc_{n+1} 
    } , \\
    \myc_n := & \, (1-\kappa)^2 \myc_{n+1} - \frac{(1-\kappa)^2 (1 - \myb_{n+1} + 2 \eta \myc_{n+1} )^2}{2 \eta + 4 \mynu + 4 \mya_{n+1} - 4 \eta \myb_{n+1} + 4 \eta^2 \myc_{n+1}}. \end{aligned}\right. 
\end{equation}
Then the  value function in~\eqref{def:valuefunctionUC} is quadratic in its arguments and given by
\begin{equation} \label{prop:benchmark:valuefunction}
V^\circ_n(x, d) = \mya_n x^2 + \myb_n x d + \myc_n d^2 + \sigma^2 \sum_{j=n+1}^N \myc_{j} \quad (n=1,\ldots,N).
\end{equation}
Moreover, the optimal
order execution strategy $(\myx^\circ_n)_{n=1,\ldots,N} \in \mathcal{A}_1^\circ$ is given by
\begin{equation} \label{prop:benchmark:feedback}
\begin{aligned}
    \myx^\circ_n = & \, \frac{ ( 2 \mynu + 2 \mya_{n+1} - \eta \myb_{n+1}) X^\circ_{n-1} - (1 - \myb_{n+1} + 2 \eta\myc_{n+1}) (1-\kappa) D^\circ_{n-1}}{\eta + 2 \mynu + 2 \mya_{n+1} - 2 \eta \myb_{n+1} + 2 \eta^2 \myc_{n+1}} \quad (n=1,\ldots, N-1), \\
    \myx^\circ_N = & \, X^\circ_{N-1},
\end{aligned}
\end{equation}
where $X^\circ_{n-1} = X_0 - \sum_{j=1}^{n-1} \myx^\circ_j$ and 
\begin{equation} \label{prop:benchmark:deviation}
D^\circ_{n-1} = (1-\kappa)^{n-1} d_0 + \eta \sum_{j=1}^{n-1} (1-\kappa)^{(n-1)-j} \myx^\circ_j + \sum_{j=1}^{n-1} (1-\kappa)^{(n-1)-j} \epsilon_j
\end{equation}
for all $n=1,\ldots,N+1$. 
\end{prop}

\begin{rem}\label{rem:deterministicSolution}
In the deterministic case $\sigma=0$ and setting $\nu = 0$ as well as $d_0 = 0$, the solution in Proposition~\ref{prop:benchmark} coincides with the deterministic solution presented in~\citet[Proposition 1]{ObizhaevaWang:13}.
\end{rem}

Observe that for every $n \in \{1, \ldots, N \}$ the optimal execution policy $\myx_n^\circ$ in~\eqref{prop:benchmark:feedback} is given as a deterministic linear feedback function of the controlled random state variables $X_{n-1}^\circ$ and $D_{n-1}^\circ$. Also, the coefficients of this linear function do not depend on $\sigma$. Therefore, in the deterministic version of the problem where $\sigma = 0$, the optimal policy is given by the exact same feedback law in~\eqref{prop:benchmark:feedback}. Moreover, since the state variables $X_{n-1}^\circ$ and $D_{n-1}^\circ$ are themselves linear in $(\myx^\circ_j)_{j=1,\ldots,n-1}$ and the i.i.d.~zero-mean Gaussian noise $(\epsilon_j)_{j=1,\ldots,n-1}$ , it follows that the optimal trades $(\myx^\circ_n)_{n=1,\ldots,N}$ of the stochastic version of the problem with $\sigma >0$ are in fact just a Gaussian-distributed perturbation from the corresponding deterministic solution with mean given by the latter. 

\begin{cor} \label{lem:benchmark}
On average, 
the optimal order executions $(\myx_n^\circ)_{n=1,\ldots,N}$ in~\eqref{prop:benchmark:feedback} coincide with $\myx_n^\circ$ for the deterministic problem with $\sigma=0$. Moreover, viewed as a random variable, $\myx_n^\circ$ has a Gaussian distribution with the above mean.  
\end{cor}

\begin{rem}[Profitable Round-Trip Strategies]
Let $N \geq 2$. By virtue of Lemma~\ref{lem:positive} below, we obtain that the coefficients $(c_n)_{n=1,\ldots,N-1}$ in~\eqref{prop:benchmark:recursion} are all strictly negative for $\kappa < 1$. As a consequence, if $\sigma > 0$ and taking zero initial buying volume outstanding ($X_0 = 0$), we have $V^\circ_n(0, d) < V^\circ_n(0,0) < 0$ in~\eqref{prop:benchmark:valuefunction} for $n<N$. In other words, the resulting optimal strategy in~\eqref{prop:benchmark:feedback} is a profitable round-trip strategy in the sense  of~\citet{HubermanStanzl:04}. This happens because the feedback policy readily exploits a nonzero deviation and its resilience as it arises; see also the detailed discussion in~\cite{AckermannKruseUrusov:21_1,FruthSchoenebornUrusov:14, FruthSchoenebornUrusov:19}. Similarly, in line with the latter references, observe that there are no profitable round-trip strategies (nor transaction triggered price manipulation strategies in the sense of~\citet{AlfonsiSchiedSlynko:12}) in the deterministic resilience case $\sigma = 0$, as long as the initial deviation satisfies $d_0 = 0$ (because $V^\circ_0(0,0) = 0$ and $u^\circ_n \equiv 0$ for all $n=1,\ldots,N$).
\end{rem}

\subsection{Deterministic Resilience Case $\sigma=0$}

In the deterministic case $\sigma=0$, the solution presented in Proposition~\ref{prop:benchmark} can be rewritten in an explicit closed-form without backward recursion in~\eqref{prop:benchmark:recursion} and forward feedback policy in~\eqref{prop:benchmark:feedback}. In view of Corollary~\ref{lem:benchmark}, this is useful for revealing the dependence of the average optimal order executions $\myx^\circ_1, \ldots, \myx^\circ_N$ on the model parameters $\kappa \in (0,1]$, $\eta > 0$ and $\nu \geq 0$. 

\begin{prop} \label{prop:benchmark_det}
Let $\sigma =0$. Moreover, let $\mathfrak{a}, \mathfrak{b}, \mathfrak{c}, \mathfrak{a}^x, \mathfrak{b}^x, \mathfrak{c}^x, \mathfrak{a}^d, \mathfrak{b}^d, \mathfrak{c}^d \in \mathbb{R}^N$ denote the vectors defined in~\eqref{def:mathcalabc}, \eqref{def:mathcalabcx}, \eqref{def:mathcalabcd} below. Set
\begin{equation} \label{prop:benchmark_det_bc}
\begin{aligned}
\hat{b} := - \frac{\mathfrak{a}^{\top} (\eta\mathfrak{b}+(1-\kappa)\mathfrak{b}^d) + (1-\kappa) (\mathfrak{a}^d)^\top \mathfrak{b} + 2 \nu (\mathfrak{a}^x)^\top \mathfrak{b}^x}{\eta \mathfrak{a}^\top \mathfrak{a} + 2 (1-\kappa) (\mathfrak{a}^d)^\top \mathfrak{a} + 2 \nu (\mathfrak{a}^x)^\top \mathfrak{a}^x}, \\
\hat{c} := - \frac{\mathfrak{a}^{\top} (\eta\mathfrak{c}+(1-\kappa)\mathfrak{c}^d) + (1-\kappa) (\mathfrak{a}^d)^\top \mathfrak{c} + 2 \nu (\mathfrak{a}^x)^\top \mathfrak{c}^x}{\eta \mathfrak{a}^\top \mathfrak{a} + 2 (1-\kappa) (\mathfrak{a}^d)^\top \mathfrak{a} + 2 \nu (\mathfrak{a}^x)^\top \mathfrak{a}^x}.
\end{aligned}    
\end{equation}
Then the optimal order execution strategy from Proposition~\ref{prop:benchmark} in~\eqref{prop:benchmark:feedback} is given by
\begin{equation} \label{prop:benchmark_det_x1}
\myx^\circ_1 = \hat{b} \cdot d_0 + \hat{c} \cdot X_0
\end{equation}
and
\begin{equation} \label{prop:benchmark_det_xi}
\myx^\circ_n = \mathfrak{a}_n \cdot \myx^\circ_1 + \mathfrak{b}_n \cdot d_0 + \mathfrak{c}_n \cdot X_0 \qquad (n=2,\ldots,N).
\end{equation}
Moreover, the optimally controlled deviation process and remaining inventory are given by 
\begin{equation} \label{prop:benchmark_det_XD}
\begin{aligned}
X^\circ_n = & \, \mathfrak{a}^x_{n+1} \cdot \myx^\circ_1 + \mathfrak{b}^x_{n+1} \cdot d_0 + \mathfrak{c}^x_{n+1} \cdot X_0, \\ D^\circ_n = & \, \mathfrak{a}^d_{n+1} \cdot u^\circ_1 + \mathfrak{b}^d_{n+1} \cdot d_0 + \mathfrak{c}^d_{n+1} \cdot X_0
\end{aligned}
\qquad (n=0,\ldots,N-1).
\end{equation}
\end{prop}

Observe that in Proposition~\ref{prop:benchmark_det} the deterministic optimal strategy $u^\circ_1,\ldots,u^\circ_N$ in~\eqref{prop:benchmark_det_xi} is fully characterized by the first trade $u^\circ_1$ in~\eqref{prop:benchmark_det_x1} and the size of the orders vary over time $n \in \{1,\ldots,N\}$. A further simplification is obtained when there is no urgency/inventory penalty, i.e., $\nu=0$. Specifically, all intermediate trades $u^\circ_2,\ldots,u^\circ_{N-1}$ are flat and determined as a $\kappa$-fraction of the first order $u^\circ_1$, shifted by a proportion of $d_0$. 

\begin{cor} \label{cor:benchmark_det}
Set $\nu = 0$ in Proposition~\ref{prop:benchmark_det}. Then in~\eqref{prop:benchmark_det_bc} it holds that
\begin{equation} \label{cor:benchmark_det_bc}
    \hat{b} = - \frac{\eta b (N-2) \big(1 + \kappa (N-1) \big) + \kappa (1-\kappa)}{\kappa\eta(N-1)(2+(N-2)\kappa)}
    \quad \text{and} \quad \hat{c} =  \frac{1}{2+(N-2)\kappa}, 
\end{equation}
and the optimal intermediate trades in~\eqref{prop:benchmark_det_xi} with initial trade $\myx^\circ_1$ in~\eqref{prop:benchmark_det_x1} are constant and simplify to
\begin{equation} \label{cor:benchmark_det_opt_1}
\myx^\circ_n = \kappa \cdot \myx^\circ_1 + \tilde{b} \cdot d_0 \qquad (n=2,\ldots,N-1)
\end{equation}
with $\tilde{b}$ defined in~\eqref{def:constants} below. The final trade is given by $\myx^\circ_N = X_0 - (1+(N-2)\kappa) \cdot \myx^\circ_1 - (N-2) \tilde{b} \cdot d_0$. The optimally controlled deviation process and remaining inventory in~\eqref{prop:benchmark_det_XD} simplify to
\begin{equation} \label{cor:benchmark_det_xd}
\begin{aligned}
X^\circ_n = & \, X_0  - (1+(n-1)\kappa) \cdot \myx^\circ_1 - (n-1) \tilde{b} \cdot d_0, \\ D^\circ_n = & \, \eta \cdot \myx^\circ_1 + (1-\kappa) \cdot d_0
\end{aligned}
\qquad (n=1,\ldots,N-1).
\end{equation}
In particular, if $d_0 = 0$ it holds that
\begin{equation} \label{cor:benchmark_det_opt_2}
    \myx^\circ_1 = \myx^\circ_N = \frac{X_0}{2+(N-2) \kappa} \quad \text{and} \quad \myx^\circ_2=\ldots= \myx^\circ_{N-1} = \kappa \myx^\circ_1. 
\end{equation}
\end{cor}

Corollary~\ref{cor:benchmark_det} shows that in the deterministic case without urgency penalty $\nu=0$ the optimal execution strategy in~\eqref{cor:benchmark_det_opt_1} is constant for the intermediate trades from $2$ to $N-1$ and keeps the deviation process in~\eqref{cor:benchmark_det_xd} flat until $N-1$. If, in addition, $d_0=0$ (i.e., the setup in~\citet{ObizhaevaWang:13})
the optimal strategy simplifies to the symmetric U-shaped~\eqref{cor:benchmark_det_opt_2} and becomes independent of the instantaneous price impact parameter $\eta$ (as observed in~\cite{ObizhaevaWang:13}). Initial and last trades are the same and all remaining intermediate trades are simply prescribed as a constant $\kappa$-fraction of the initial trade. In particular, for full resilience $\kappa=1$ all trades are just equal to $X_0/N$.

\begin{rem}
The simple formula in~\eqref{cor:benchmark_det_opt_2} has also been derived in~\citet[Corollary 6.1]{AlfonsiFruthSchied:10}.
\end{rem}

\subsection{Optimal Execution Profiles}

Figure~\ref{fig:benchmark_control} illustrates the behaviour of an optimal buying program from Proposition~\ref{prop:benchmark_det} with linear transient price impact ($\alpha=1$) for different values of the model parameters $\kappa, \eta$. We consider buying  $X_0 = 100,000$ shares in $N=10$ trades, with initial deviation $d_0 = 0$. As discussed, when $\nu=0$ we obtain the well-known U-shape for the buying schedule, driven only by $\kappa$. With positive $\nu > 0$ several novel qualitative effects appear:
 (i) the optimal strategy now depends on \emph{both} resilience rate $\kappa$ \emph{and} temporary price impact $\eta$; (ii) intermediate trades are in general not flat anymore; and (iii) the overall pattern may shift from a U-shaped to a monotonically decreasing sequence of trades. The presence of an inventory penalty $\nu > 0$ creates an incentive to buy faster in the beginning, which becomes more pronounced for small $\eta$. Consequently, the case of $\nu >0$ and small $\eta$ might lead to a decreasing execution schedule, cf.~the purple curve in Figure~\ref{fig:benchmark_control}.

\begin{figure}[!htb]
\centering
\includegraphics[height=2in,trim=0.1in 0.25in 0.2in 0in]{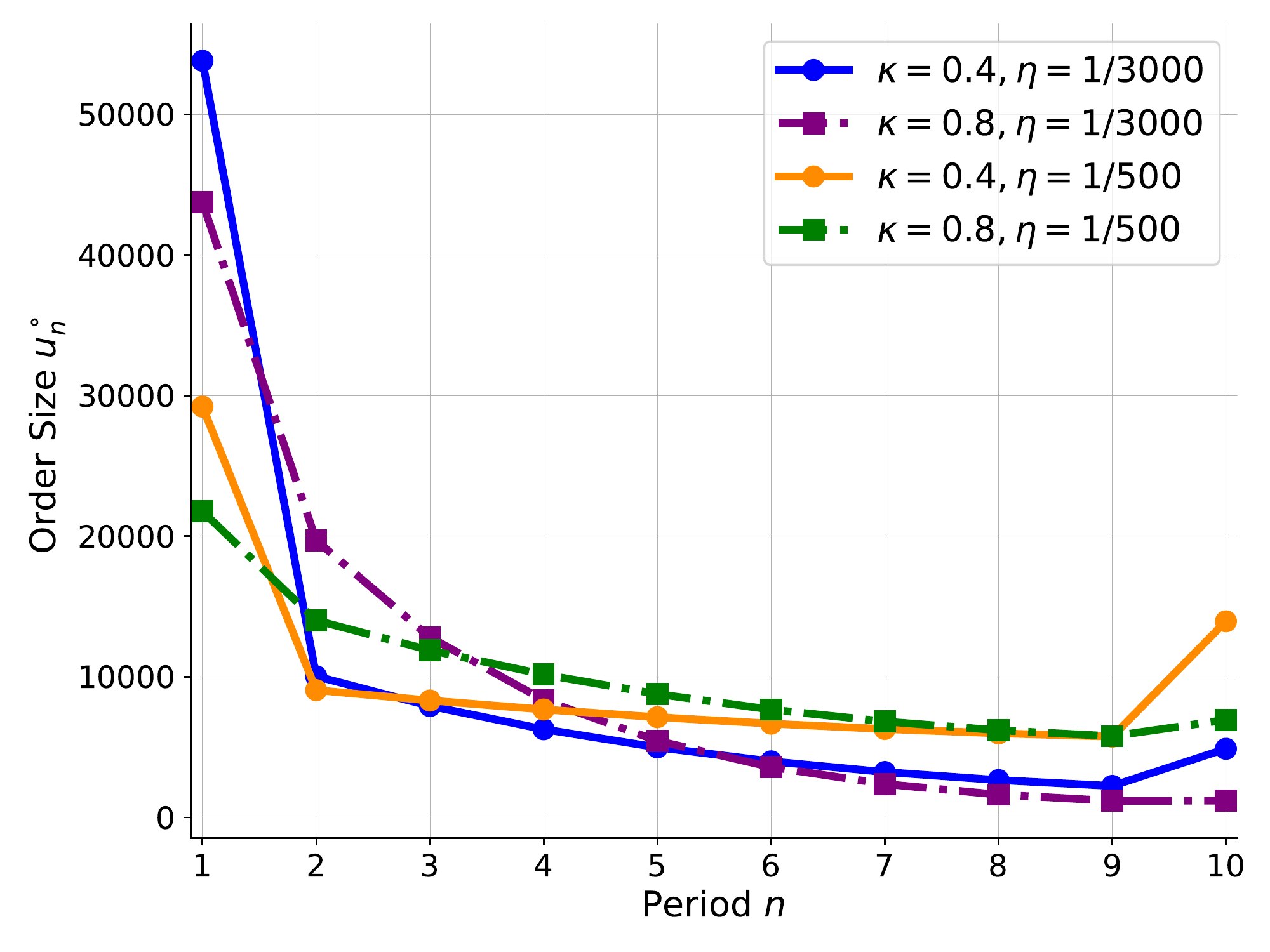} 
\caption{ Benchmark unconstrained buying programs $(u^\circ_n)$ based on  Proposition~\ref{prop:benchmark_det} for the case $\sigma = 0$. We consider executing $X_0 = 10^5$ shares in $N=10$ steps for different values of the model parameters $\kappa, \eta$ and $\nu = 0.00005$, $d_0=0$.   \label{fig:benchmark_control}}
\end{figure}
 
  Observe that relatively small changes in model parameters generate significant impact on the optimal strategy $\myx^\circ_n$. For example, the initial trade $\myx^\circ_1$ can be over 55\% of total $X_0$ when $\kappa$ is small and $\eta$ is small (which causes the emphasis to be on inventory risk), and less than 25\% of $X_0$ for $\kappa$ large and $\eta$ large (where strong resilience and strong temporary impact encourage to trade nearly equal amounts at each step). Also note the non-monotone behavior of the strategies, where the curves $n \mapsto \myx^\circ_n$ cross each other at different steps.
 
In the right panel of Figure~\ref{fig:benchmark_boxplot} we illustrate how the pathwise strategies $\myx_n^\circ$ for $\sigma >0$ from Proposition~\ref{prop:benchmark} are normally distributed and coincide on average with the deterministic solution, cf.~Corollary \ref{lem:benchmark}. In the plot the blue bars correspond to the optimal deterministic solution from Proposition~\ref{prop:benchmark_det} and Corollary~\ref{cor:benchmark_det}; and the boxplots show the distribution of $(\myx_n^\circ)_{n=1,\ldots,10}$ from Proposition~\ref{prop:benchmark}. We take $\sigma = 2$ and average over 10,000 paths, considering both the classical case with $\nu=0$  (left panel with U-shaped strategy) and our extension to positive inventory penalty $\nu >0$ (right panel). In line with Corollary~\ref{lem:benchmark}, the empirical means match the deterministic solution. Since this feature is true for any configuration of $(\kappa, \eta, \nu)$ the above qualitative effects can be deduced for the stochastic optimal feedback controls in~\eqref{prop:benchmark:feedback}.

\begin{figure}[!htb]
\centering
\begin{tabular}{cc}\hspace*{-20pt}
\includegraphics[height=1.75in]{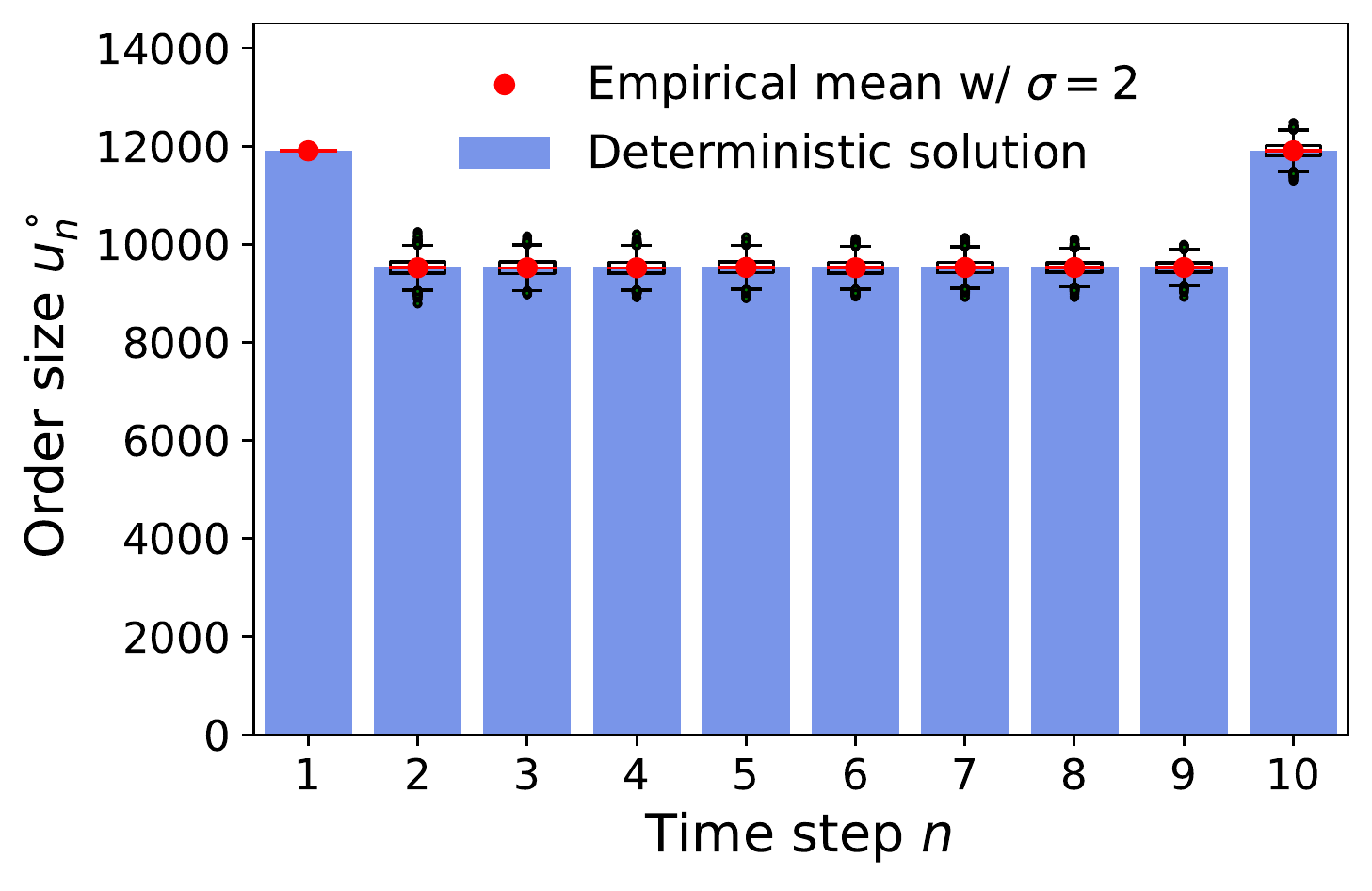} &
\includegraphics[height=1.75in]{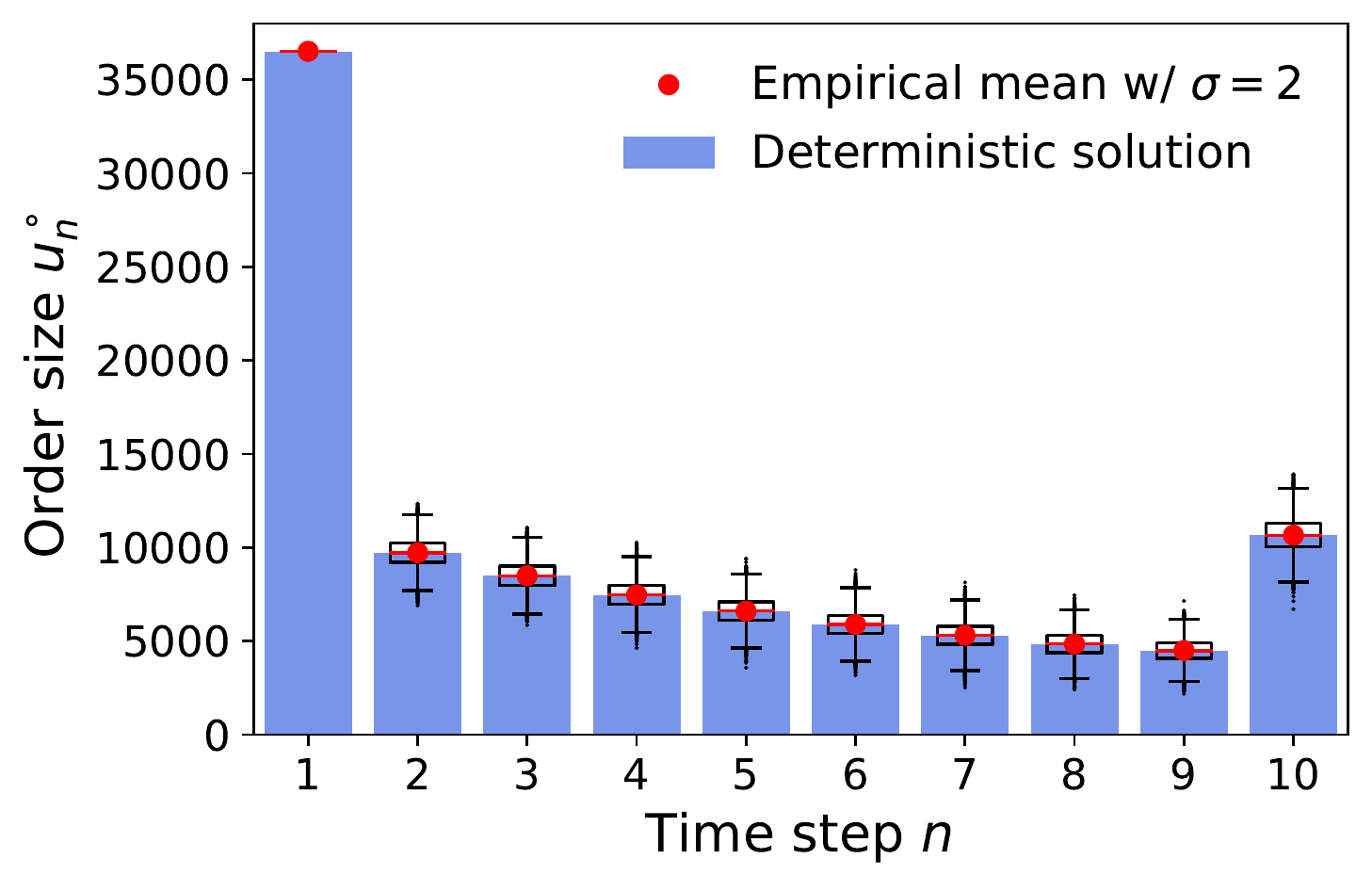} \\
 $\kappa = 0.8, \eta = 1/500, \nu = 0$ & $\kappa = 0.4, \eta = 1/1000, \nu = 0.00005$ \\
\end{tabular}
\caption{ Distribution of $(\myx_n^\circ)_{n=1,\ldots,10}$ from Proposition~\ref{prop:benchmark} over 10,000 paths for $\sigma = 2$ (boxplots) together with the empirical mean (in red) compared to the  formula from Proposition~\ref{prop:benchmark_det} for $\sigma=0$ (barplot).  \label{fig:benchmark_boxplot}}
\end{figure}

\section{Numerical Implementation}\label{sec:method}

In the case of nonlinear transient price impact $\alpha \neq 1$, no closed-form solution is possible and numerical methods are needed to compute an optimal solution $(u^*_n)_{n=1,\ldots,N}$ for the minimization problem in~\eqref{def:optProblem}. A numeric algorithm is also needed for $\alpha=1$ to handle the constrained buy-only setting $u_n \ge 0$. Finally, the algorithm will be useful to study the dependence of $V_n,u^*_n$ on model parameters. Our algorithm to compute the value function and the optimal strategy  relies on the Bellman equation \eqref{eq:bellman-2} which provides a recursive characterization for $V_n$ and therefore for $\myx^*_n$. Given the parametric setup, we consider a generic state $y$ that fuses stochastic states $(x,d)$ and some (or none) of the aforementioned model parameters; and for the rest of this section treat $V_n$ and $u_n^*$ as a function of $y$.

In order to solve \eqref{eq:bellman-2} we need to be able to evaluate its right-hand side. This entails (i) evaluating $V_{n+1}(\cdot)$; (ii) evaluating the conditional expectation; (iii) taking the $\inf$ over $\myx$. None of those steps are possible to do analytically and numerical techniques are necessary. 

Below, we propose and implement a direct approach based on constructing a functional approximator, also known as a \emph{surrogate}, $\hat{V}$, to the value function that is trained via an empirical regression. This method is also known as \emph{Approximate Dynamic Programming (ADP)} and \emph{Projected Value Iteration (PVI)} in the literature; see, e.g., \citet{KeshavarzBoyd:12} and the references therein. Specifically, we consider the use of (feed-forward) Neural Networks (NN) for the latter. Note that our method is based on the DP equation; we do not make any reference to Hamilton-Jacobi-Bellman equations that are used in continuous-time setups, and which admit their own suites of NN approaches. Neither do we consider reinforcement learning (RL) that dispenses with the divide-and-conquer paradigm underlying the Bellman equation and the value function, and aims to maximize total trading revenue on the entire horizon. Namely, RL solves for all $V_n$'s in parallel over $n=1,2,\ldots,N$ while we maintain the sequential backward learning of $V_{N-1}, V_{N-2}, \ldots, V_1$. 

Instead, our approach is conceptually faithful to the classical DP philosophy and offers the following advantages:
\begin{itemize}
    \item It is straightforward to understand and implement, in some sense offering the most immediate approach to employing surrogate and other machine learning techniques for DP. As such, we bypass the more subtle ideas that have been advanced, offering a pedagogic-flavored setup;
    
    \item It provides modularization, emphasizing that NN solvers are just one type of \emph{many} potential surrogates. Thus, it de-mystifies deep learning, simply treating it as a choice among many. Indeed, our implementation requires just a few lines of code to substitute a different surrogate type.  
\end{itemize}

\subsection{Building a Surrogate}

Denote by $\mathcal{Y} \subseteq \mathbb{R}^\ell$ the state space of the surrogate, which includes the bona fide inputs $x,d$, as well as all the relevant model parameters among $\kappa, \eta, \alpha, \nu$ which we wish to capture. In the examples below we consider $\ell \in \{2,3,4,5\}$, with the main illustrative example being $\ell=4$ where we take $y=(x,d,\kappa,\eta)$, i.e., we simultaneously learn the value function and the strategy as a function of $(x,d)$, the resilience $\kappa$ and the instantaneous price impact $\eta$ for some fixed $\alpha$ and $\nu$. 

We denote by $\mathbf{G}(\my,\myx, \epsilon)$ the one-step transition function of $y$ given external noise $\epsilon$ and action~$\myx$. For example, 
$$
\mathbf{G}( (x,d,\kappa,\eta),\myx, \epsilon) = (x-\myx, (1-\kappa)d + \eta  \myx^{\alpha} + \epsilon, \kappa, \eta).
$$

Our resolution of the DP equation~\eqref{eq:bellman-2} operates as-is with each of the underlying sub-steps. This means that considering a generic intermediate step $n$ of the backward recursion and given a surrogate $\hat{V}_{n+1}$ we first substitute it in place of $V_{n+1}$.  Next, the expectation in \eqref{eq:bellman-2} is over the stochastic shocks $\epsilon_{n}$ which are one-dimensional Gaussian random variables. We employ Gaussian quadrature to replace the respective integral with a finite sum. Specifically, we rely on the optimal quantization of \citet{PagesPhamPrintems:04,PagesPrintems:05} to select $j=1,\ldots,N'$ weights $w^j$ and respective knots $e^j$ to approximate 
\begin{align}\label{eq:quantize}
\mathbb{E} \left[ \hat{V}_{n+1}(\mathbf{G}(\my,\myx,\epsilon_{n})) \,\big| \, Y_{n} = y\right]  \simeq \sum_{j=1}^{N'}
w^j \hat{V}_{n+1} \left(\mathbf{G}(\my,\myx,e^j) \right). 
\end{align}

\begin{rem}
One can straightforwardly consider non-Gaussian (e.g., heavy tailed) noise distributions for $\epsilon_n$ thereby providing a different nonlinear generalization. Non-Gaussian $\epsilon_n$ just reduces to taking a different set of $e^j,w^j$'s.
\end{rem}

The optimization over the number of shares to buy $\myx_n$ is done via a numerical optimizer, namely the standard gradient-free optimization routine (such as L-BFGS) that is available in any software package.  Note that $\myx_n$ is scalar, allowing the use of fast one-dimensional root finding algorithms. In our implementation, we do not evaluate any gradients of $\hat{V}_{n+1}$ although that is feasible. In order to restrict to $\myx \in [0,x]$ and rule out any selling, we optimize on the above bounded interval, straightforwardly supported by such solvers.

Finally, it remains to construct $\hat{V}_n$. As a machine learning task, the goal is to learn the true input-output map $\my \mapsto V_{n}(\my)$. Such functional approximation, aka surrogate construction \cite{Gramacy:20}, is carried out by  selecting a collection of training inputs $\my^{1:M} \in \mathcal{Y}$, evaluating (a noisy version of) $V_{n}(\my^{1:M}) =: v^{1:M}_n$ and then fitting a statistical representation $\hat{V}_n(\cdot)$ that can interpolate to new, out-of-sample $\my$'s. The evaluation of $v^m_n$ is achieved by direct computation via the nonlinear optimizer and the quantized integral on the right-hand-side of \eqref{eq:bellman-2}:
\begin{equation}
\begin{aligned} \label{eq:vn-pointwise}
v_n^m := & \, \inf_{\myx \in [0,x^m_{n-1}]} \left\{(1-\kappa^m) d^m_{n-1} \myx + \frac{\eta^m}{2} \myx^{\alpha^m+1} + \nu^m(x^m_{n-1}-\myx)^2 \right. \\ \, & \hspace{52pt} \Bigg. + \sum_{j=1}^{N'}
w^j \hat{V}_{n+1}(\mathbf{G}(\my^m_{n-1},\myx,e^j)) \bigg\}.
\end{aligned}
\end{equation}
Observe that \eqref{eq:vn-pointwise} uses $\hat{V}_{n+1}$ and so leads to a recursive construction backward in time, fully mimicking the dynamic programming equation. This recursion is instantiated with the exact terminal condition $\hat{V}_{N}(\my) \equiv V_N(\my) = \left((1-\kappa)d + \frac{\eta}{2} x^\alpha\right) x$ and then run for $n=N-1,N-2, \ldots, 1$. 
Our notation emphasizes the pointwise nature of the optimization for $v^m_n$ by taking the static parameters $\kappa,\eta,\alpha,\nu$ as part of the training design, hence varying in $m$.

\subsection{Policy Approximation}
The primary output of the numerical solver is the \emph{execution strategy}, given in feedback form as $\my \mapsto \myx^*_n(\my)$. Indeed, the strategy is what the controller is ultimately after, and yields a clear interpretation of how many shares the solver recommends to buy next. In contrast, the approximate value function $\hat{V}$ is harder to interpret (since it is not in pure monetary dollars but potentially also involves the abstract inventory costs) and moreover due to the intermediate approximations does not have any concrete probabilistic representation. Indeed, while the true $V$ is the expected cost of the strategy, $\hat{V}$ is \emph{not} an expectation on $[0,T]$, since it is obtained from one-step recursions.

Conventionally, $\myx^*$ is characterized as the $\arg\inf$ of the Bellman recursion \eqref{eq:bellman-2}, so that to obtain $\myx^*_n(\my)$ one must re-do the optimization over $\hat{V}_n$. This is time-consuming, inefficient  and non-transparent to the user. Instead we seek a direct representation for $\myx^*_n$ and in the spirit of the machine learning mindset (specifically actor-critic frameworks) propose to construct a second surrogate, auxiliary to the one describing $\hat{V}$. Accordingly, we construct a separate surrogate $\hat{\myx}_n(\cdot)$ that is trained based on the recorded $\myx^m_n$, the optimal execution amounts for each $\my^m_{n-1}$. Note that the fitting of $\hat{\myx}_n$ is independent of the main loop above, so can be done in parallel with $\hat{V}_n$ or after-the-fact. 

Even when the training inputs $\myx^m_n$ are in the range $[0,x^m_{n-1}]$, it is not generally guaranteed that this would be true for the fitted prediction $\hat{u}(y^m_{n-1})$. In order to enforce this constraint, we train $\hat{u}$ based on the \emph{fractions} $\myx^m_n/x^m_{n-1} \in [0,1]$ and use a \emph{sigmoid} transformation to intrinsically restrict all predictions (again interpreted as fractions of current inventory to be bought) into $(0,1)$.

\begin{rem}
In our setup, we first pointwise approximate $\myx_n(y^m_{n-1})$ and then fit a statistical surrogate to those samples; in~\citet{Pham18deep1,Pham18deep2} the strategy is the opposite: first parametrize potential $\myx(\cdot; \theta_n)$ through a neural network with weights $\theta_n$; then use back-propagation to optimize the respective hyper-parameters $\theta_n$. In that sense, their resulting $y \mapsto \myx(y; \theta_n)$ is not a solution of any optimization problem, and there is no underlying dataset $(\my^{1:M}_{n-1}, u^{1:M}_n)$ like for our approach.
\end{rem}

\subsection{Neural Network Solvers}

A popular class of surrogates consists of feed-forward neural networks (NN). NN's allow efficient fitting of high-dimensional parametric surrogates using back-propagation and a variety of stochastic optimization techniques, such as stochastic gradient descent.

 An NN represents $\hat{V}_n$ as a composition, using linear hidden units at each layer, a user-chosen activation function across layers, and a user-selected number of layers. In our context, the precise architecture of the feed-forward NN is not conceptually important. The NN parameters are optimized via batch stochastic gradient descent. Relative to other statistical models, NN is \emph{overparametrized} with thousands of parameters, known as the weights. Nevertheless NN is known to enjoy excellent empirical convergence (i.e., the algorithms find near-optimal weights), especially for large scale datasets, including in other stochastic control applications \cite{Pham18deep2,han2016deep,IsmailPham:19r,Pham18deep1}.

Training an NN requires specifying the training inputs $\my^{1:M}_n \equiv (x^{1:M}_n, d^{1:M}_n, \kappa^{1:M}_n, \eta^{1:M}_n) \in \mathcal{Y}$ and initializing the NN weights. For the former, we propose a space-filling experimental design, matching the standard approach in statistics. As default,
we select a hyper-rectangular training domain $\bar{\mathcal{Y}}$ and 
 sample each coordinate of $\my$ uniformly and independently on the respective training interval, for example we sample the inventory $x_n \in [0, X_0]$.  One may also implement joint sampling in $\bar{\mathcal{Y}}$, such as Latin Hypercube Sampling (LHS). A further option is to use low-discrepancy Quasi Monte Carlo (QMC) sequences to achieve a space filling training set of arbitrary size $M$. Compared to i.i.d.~Uniform sampling, LHS and QMC offer lower variance and better coverage, avoiding any clusters or gaps in the training locations, which is relevant when $M$ is relatively small compared to the dimension of $\my$. Note that training sets are indexed by $n$; we sample fresh $\my^{1:M}_n$'s at each step, so that the training inputs vary across $n$'s, although they have the same size and shape.

To initialize the weights, we found it very beneficial to \emph{rescale} the inputs to the unit hypercube, which permits the use of standard NN weight priors (namely Truncated Normal with mean zero). Similarly, for $\hat{V}$ we re-scale the training outputs $v^m_n$ to be in the range $[0,1]$ too. For training the policy surrogate $\hat{\myx}$ we apply a sigmoid activation function on the output layer of the NN that directly ensures that $\hat{\myx}(\my) \in (0,1)$. The latter fraction is multiplied by the current inventory $x$ to get the number of shares to trade.
 
Our implementation (see the supplementary Jupyter Notebook) employs the \texttt{TensorFlow} library in \texttt{Python}, which provides one of the most popular engines for NNs. We employ ``factory defaults'' to construct our neural networks using the \texttt{tensorflow.keras.Sequential} architecture. This is a linear stack of layers; we use the same number of neurons per layer and the same activation function across layers. For the experiments below we utilize 4 layers and 20 neurons, with the ELU activation function, $ELU(x)= x 1_{\{x >0\}} + (e^x-1)1_{ \{x \le 0\} }$. Training uses the Adam algorithm with default learning rate, batch size of 64 and $E$ epochs. The latter represents a single pass through all training data, and many epochs are needed given the high noise in the underlying stochastic gradient descent optimizer of the NN weights.  With the above choices,  training a NN takes just a few lines in \texttt{TensorFlow}. Indeed, it takes more code to scale/re-scale the inputs and outputs than to actually build and fit the Neural Net. 
Other neural network libraries, such as \texttt{scikit-learn} or \texttt{PyTorch} could be straightforwardly substituted.

We end this section with a few final remarks:
\begin{itemize}
    \item It is completely straightforward to modify the NN architecture. For instance, to replace a 4-layer ``deep" architecture with a single layer, it takes just commenting out 3 lines in our code. Similarly, one can add more layers with a single line change.
    
    \item Because we maintain the underlying Dynamic Programming paradigm, the NNs are fitted one-by-one. Therefore, one has complete flexibility in modifying any aspect of the fitting procedure to make it step-dependent. This includes the size $M$ and shape of the training set $\my^{1:M}_{n-1}$; the parameters for the NN optimization, including the initial NN weights; the NN architecture, such as the number of neurons; and even the surrogate type itself. For example, one could mimic RL techniques to select time-dependent training regions $\bar{\cX}_n$ to reflect the natural time-dependency of the solution, such as the remaining order to fill $X_n$ decreasing over time.
    
   \item The convergence of the stochastic gradient descent to find a good $\hat{V}_n$ is quite slow. We find that $ E\gg 1000$ epochs are necessary to achieve good results. The number of epochs is the primary determinant of fit quality, cf.~Section \ref{sec:nn-config}.
   
   \item Any surrogate can be re-trained/updated at any point of the overall backward loop. For example, we suggest training $\hat{V}_n$'s for all $n$'s and then training more (i.e., run the backward loop again with same training samples or newly
generated ones) as one way to improve empirical convergence. This process is completely transparent and just requires loading the existing NN objects corresponding to $\hat{V}_n$'s
rather than initializing new ones. Similarly, one can straightforwardly implement \emph{warm starting}, using the fitted NN weights at step $n+1$ as an initial guess for the weights of $\hat{V}_n$.

\item A typical failure point for surrogates is unstable prediction when extrapolating beyond the range of the training region. Thus, care must be taken to select the training domain $\bar{\mathcal{Y}}$ in order to minimize extrapolation. Usually the modeler knows a priori the test cases of interest and so can ensure that the training range is at least as large. For example, below we wish to test for $\kappa \in \{0.4, 0.6, 0.8\}$ and therefore we train on $\kappa^{1:M}_n \in [0.38, 0.82]$ to mitigate any issues with prediction at or beyond the edge of the training domain.

\end{itemize}

\begin{rem}
Other surrogate types could be employed in place of NNs. For example,  Gaussian Processes (GP)  \cite{RasmussenWiliams:06} is a kernel regression method where the kernel hyperparameters are fitted using maximum likelihood. They are implemented in, e.g., \texttt{scikit-learn} in Python. GPs offer variable selection through  automatic relevance determination, which allows to ``turn off'' covariates/parameters that make little impact on the response. GP surrogates are known for excellent performance on limited datasets and are very popular for emulating expensive computer and stochastic experiments. Another surrogate type are LASSO linear models that explicitly project $\hat{V}_n$ onto the span
 of the basis functions $\{B_r(\my)\}$, $\hat{V}_n(\my) = \beta_0 + \sum_r \beta_r B_r(\my)$. The respective coefficients are determined from the ($L_1$-penalized) least squares equations. Finally, we note that not all regression methods are appropriate; non-smooth frameworks like Random Forests would yield unstable or discontinuous estimates of $\hat{\myx}(\cdot)$ and therefore should not be applied.
\end{rem}

\subsection{Workflow}
To summarize, the algorithmic workflow to solve the parametric optimal execution problem is given in Algorithm \ref{alg:1}. We provide a \texttt{TensorFlow} implementation of the above in the supplementary fully-reproducible Jupyter notebook.

\begin{algorithm}
\begin{algorithmic}[1]
\REQUIRE{ $M$ (number of training locations) }
\STATE Set $\hat{V}_{N}(\my) \equiv V_N(\my) = \left((1-\kappa)d + \frac{\eta}{2} x^\alpha\right) x$ (no approximation needed at terminal time)
\FOR{$ n=N - 1, N-2, \ldots, 1$ }
\STATE Initialize the neural net surrogates $\hat{V}_n$ and $\hat{\myx}_n$
\STATE Select the experimental design $\my_{n-1}^{1:M}=(x_{n-1}^m,d_{n-1}^m,\kappa^m,\eta^m,\nu^m,\alpha^m)_{m=1,\ldots,M}$
\FOR{$m=1,\ldots, M$}
\STATE Evaluate $v_n^m$ in~\eqref{eq:vn-pointwise} on $y_{n-1}^m$ using a gradient-free optimizer; record the corresponding $\myx_n^m$
\ENDFOR
\STATE Use the dataset $(\my^{1:M}_{n-1}, v^{1:M}_n)$ to fit the surrogate $\hat{V}_n(\cdot) : \mathcal{Y} \to \mathbb{R}$
\STATE Train the control surrogate $\hat{\myx}_n(\cdot)$ using $(\my^{1:M}_{n-1}, u^{1:M}_n/x^{1:M}_{n-1})$
\ENDFOR 
\STATE Generate $M'$ out-of-sample forward paths $\my^{1:M'}$
\STATE Record the cost $\check{v}^m, m=1,\ldots, M'$ based on  utilizing $\hat{\myx}_n(\my^{m}_{n-1})$ as the control at step $n$.
\RETURN Empirical average execution cost
\begin{align}\label{eq:check-v}\check{V}(0,Y_0) = \frac{1}{M'} \sum_{m=1}^{M'} \check{v}^m.
\end{align}
\end{algorithmic}
\caption{Solving the Optimal Execution Problem \label{alg:1}}
\end{algorithm}

\section{Numerical Experiments}\label{sec:results}

\subsection{Comparison to Reference Model}
With linear price impact $\alpha=1$, the unconstrained
strategy $(\myx^\circ_n)_{n=1,\ldots,N}$ is available by parsing recursively the formulas in Proposition \ref{prop:benchmark}. This yields a concrete benchmark to evaluate our numerical algorithms. In this section we compare an NN surrogate to the above ground truth. Even though $\myx^{\circ}$ is not feasible (since it can and does turn negative in some states), this comparison is still useful. For the chosen parameter configurations, $\myx^{\circ}_n \gg 0$ (cf., Fig \ref{fig:benchmark_boxplot}) so that the non-negativity constraint is not binding on the vast proportion of the paths and hence $\myx^{\circ}$ is very nearly optimal for the constrained problem as well; that is, $u^\circ \approx u^*$ with $u^*$ denoting the minimizer in~\eqref{def:optProblem}. This ``near"-optimality is confirmed by the closeness between the constrained NN-strategy $\hat{u}$ and the unconstrained $u^\circ$ strategy (LF), offering an additional consistency check on the NN solver.

To enable an apples-to-apples assessment of $(\hat{\myx}_n)_{n=0}^{N-1}$, we \emph{fix} the set of ``noise" $\epsilon^{1:M'}_n$ that feed into the realized $D_n$'s, and employ this fixed database of $\epsilon$'s to generate the test forward trajectories both for the NN approximator and for the exact solution. 
Since $D_n$ depends on the past execution amounts, the trajectories of the approximating strategy will differ at each and every step $n$. Thus, we are \emph{not} making a pairwise comparison between  $\hat{\myx}_n(\cdot)$ and $\myx^{\circ}_n(\cdot)$ from \eqref{prop:benchmark:feedback}, but compare in terms of the final execution cost in \eqref{eq:check-v}. Indeed, due to stochastic fluctuations, on any given path the realized costs from the latter may be higher or lower than the costs of the benchmark strategy, however the Law of Large Numbers guarantees that for $M'$ large, the \emph{average} execution cost must be at least as much as the benchmark.

For illustration purposes, we train an NN that takes in the four inputs $(x, d, \kappa, \eta)$ and compute the corresponding $\hat\myx_n$. The respective 4-dimensional training domain is taken to be a hyper-rectangle specified as $x \in [0,10^5], d \in [0,100], \kappa \in [0.38, 0.82], \eta \in [1/900, 1/5000]$. The remaining parameters are fixed as $\nu =0.00005, \sigma=1$. We then select $M=4000$ training points i.i.d.~uniformly in the above domain, independently for each step $n$. Note that while we kept the same training domain across steps, one can vary this as $n$ changes. The ranges of the training domain are ultimately driven by the desired \emph{test configurations}. For example, below we show the results for the test set with $\kappa=0.4, \eta = 1/1000$ and initial condition $X_0 = 10^5, D_0 = 0$. The latter affect the range of $X_n$. Since $X_n \le X_0$ and $X_N =0$, we train on the range $[0,X_0]$. The range for $D_n$ depends on $D_0, X_0$ and $\eta$. With $X_0 = 10^5$, we expect to buy up to 50-60K shares in the first step, which with $\eta \simeq 1/1000$ would lead to $D_n \in [40,80]$. Note that $D_n$ is highly sensitive to $\eta$, hence the respective training range should reflect the bounds on $\eta$ values. The range for the resilience $\kappa$ covers the test case of $\kappa=0.4$ and is simultaneously quite wide to cover a range of market conditions (calibration of $\kappa$ is known to be difficult). The range for the instantaneous impact $\eta$ is similarly chosen to cover a range of market conditions, with $\eta=1/1000$ leading to impact that is 5 times stronger than that of $\eta=1/5000$. The quantization of the conditional expectation in \eqref{eq:quantize} uses $N'=50$ knots.

\begin{figure}[!ht]
\centering
\begin{tabular}{cc}
\includegraphics[height=2in,trim=0.2in 0.1in 0.25in 0.1in]{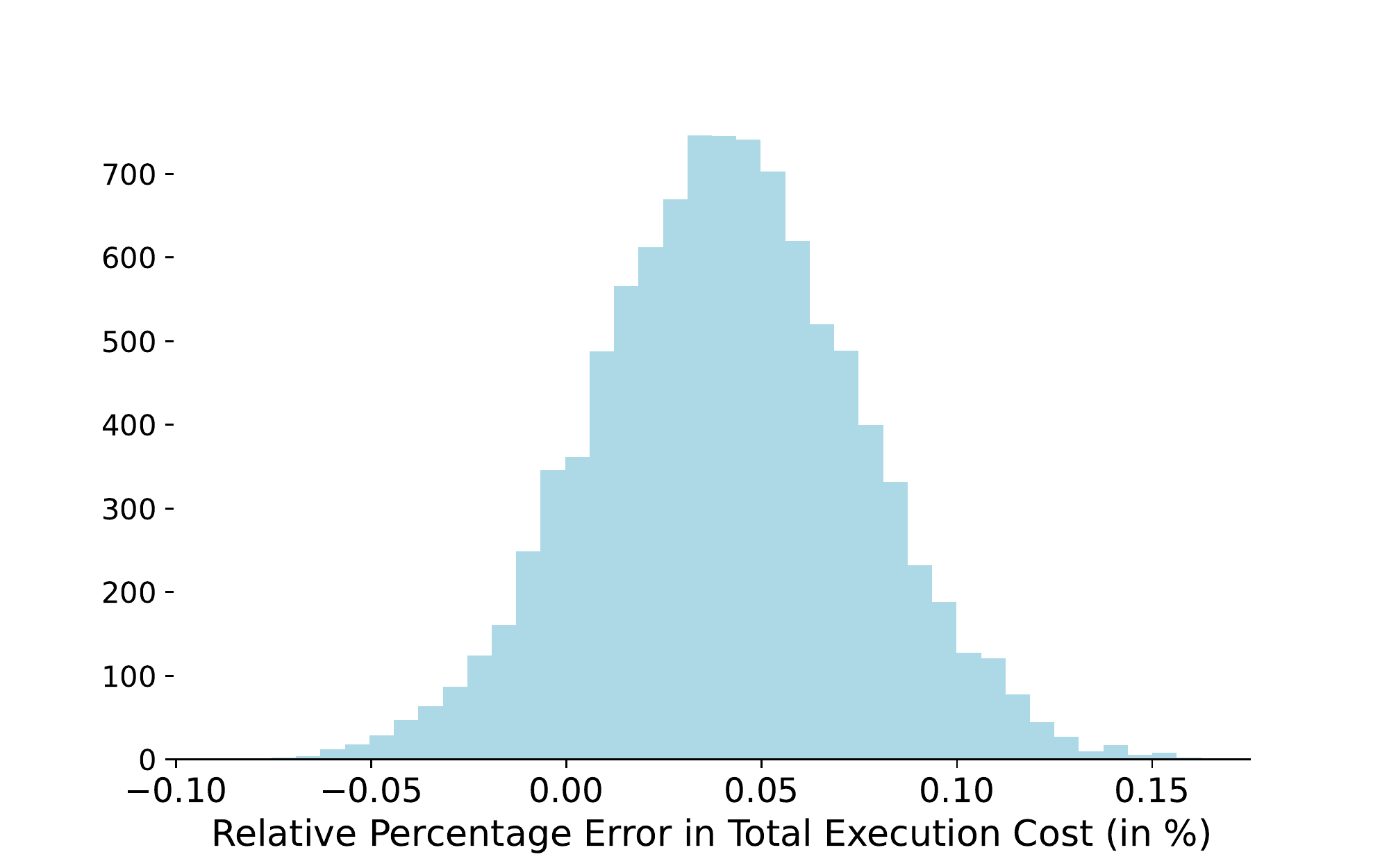} &
\includegraphics[height=2in,trim=0.25in 0.1in 0.2in 0.1in]{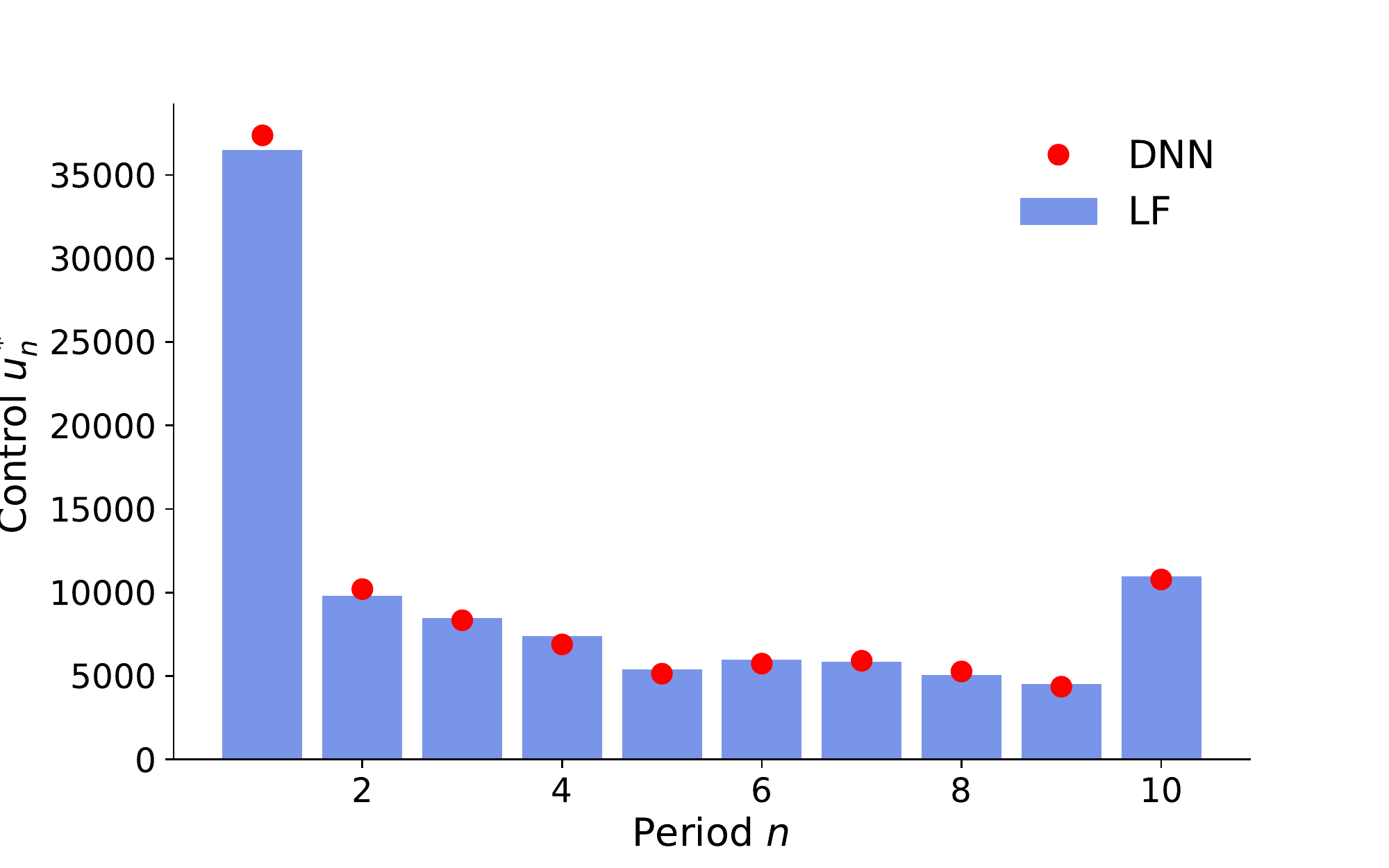}
\end{tabular}
\caption{\emph{Left:} Histogram of the relative error in final costs between the NN-strategy and the benchmark strategy based on $M'=10^4$ forward trajectories. \emph{Right:} Optimal policy vs the NN-strategy along the same path (fixed set of $\epsilon_n$'s). Parameters fixed at $\kappa=0.4,\eta=1/1000,\nu=0.00005$ and $\sigma=1$.\label{fig:nn-hist}}
\end{figure}

The left panel of Figure \ref{fig:nn-hist} shows the histogram of the difference between the realized execution costs $\check{v}^m$ coming from an NN approximator and the  benchmark, across $10^4$ test trajectories. We observe that on average the former are 0.041\% higher. 
Note that for about 11.6\% of the paths, the realized costs from our approximate strategy were \emph{less} than from the benchmark. Conversely, on more than 95.5\% of the paths, the NN-based costs were less than 0.1\% higher than those from the reference strategy, which is practically a very good level of accuracy.

The right panel of Figure~\ref{fig:nn-hist} compares the reference $u^\circ(\my^{m'}_{1:N})$ and the NN-based control $\hat{\myx}(\my^{m'}_{1:N})$ on one sample forward trajectory. We observe that the two strategies are very close in a pathwise sense as well, confirming the high approximation quality.

\subsection{Nonlinear Price Impact}

We now proceed to consider strategies with nonlinear price impact $\alpha \neq 1$. 
Two comparators are the linear feedback (LF) strategy $u^\circ$ defined by \eqref{prop:benchmark:feedback} and constrained to buys-only, which we denote by
$u^{LF}_n := u^\circ_n \vee 0 \wedge X_n$,
and the volume-weighted average policy, known as VWAP. The first $\myx^{LF}$ comparator utilizes the  linear benchmark formula as a function of current $X_n,D_n$, in other way it postulates the counter-factual $\alpha=1$ even when $\alpha$ is not unity. This means it will under-estimate price impact when  $\alpha > 1$, and will overestimate price impact when $\alpha < 1$. The above over- or under-estimation of price impact can lead to severe instability in the execution strategy. The second VWAP comparator fixes $\myx^{\text{VWAP}}_n \equiv X_0/N$. That strategy is completely model independent, and therefore its performance is little affected by $\alpha$. In that sense the VWAP is the opposite of the linear feedback policy $\myx^{LF}$ which makes strong assumptions on the market environment; VWAP implements the same trades no matter the model parameters.  

\begin{figure}[ht] 
\centering
\begin{tabular}{cc} 
\includegraphics[height=2in,trim=0.2in 0.2in 0.2in 0.5in]{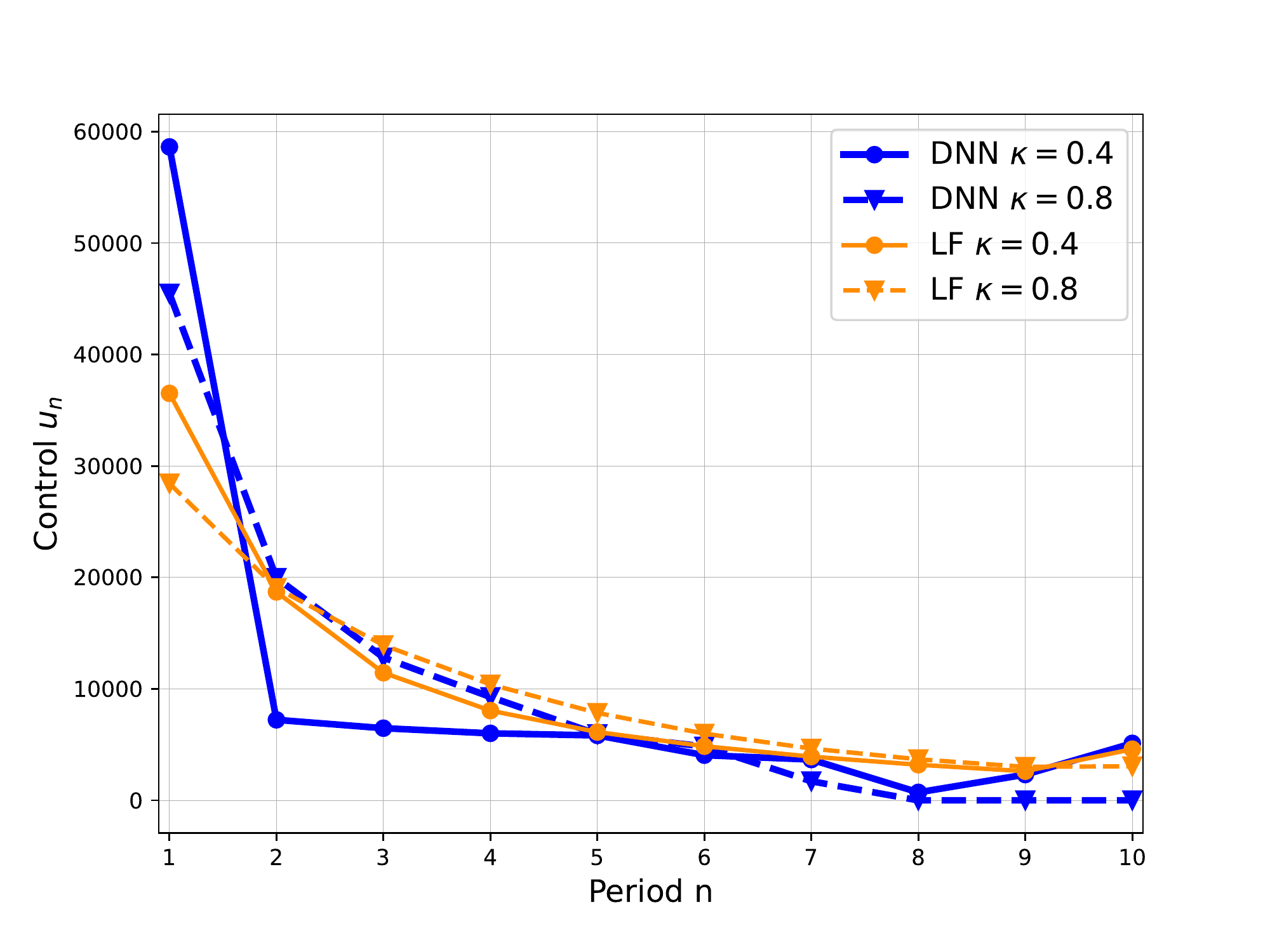} & 
\includegraphics[height=2in,trim=0.2in 0.2in 0.2in 0.5in]{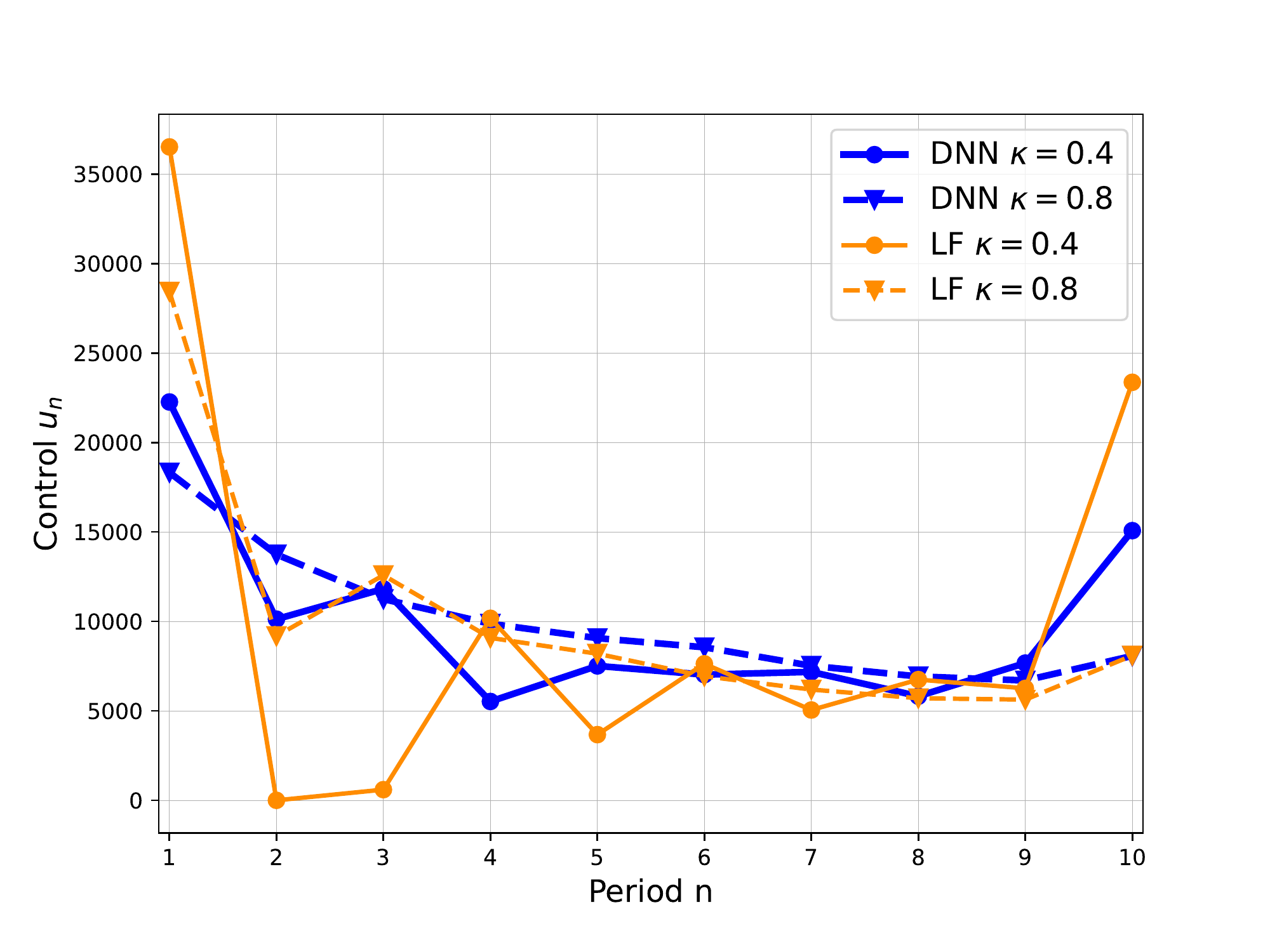}\\
$\alpha =0.9$ & $\alpha=1.1$
\end{tabular}
\caption{Execution strategy (mean values across $M'=10^4$ simulations) for $X_0 = 100,000,D_0=0$ and $N=10$ with $\sigma = 1$, urgency parameter $\nu = 0.0001$, $\eta=1/500$ and different price impact parameters $\kappa$ and $\alpha$, utilizing a 5D solver in $(x,d,\kappa,\eta,\alpha)$. 
\label{fig:feedbackCasestudies}}
\end{figure}

While for a fixed set of model
parameters there are multiple ways to construct a nonlinear solver, this would be computationally intractable to do for a large collection of $(\kappa, \eta, \alpha)$. Consequently, our \emph{parametric solver} is indispensable to provide a comprehensive solution across many parameter configurations. Moreover, since we train jointly across a range of $\alpha$, the linear case $\alpha=1$ is covered and can be used to indirectly assess the fit quality. In other words, we expect the performance of the NN surrogate for $\alpha \neq 1$ to be \emph{similar} to its performance when $\alpha =1$.

Figure~\ref{fig:feedbackCasestudies} compares the resulting NN-based strategies to the linear feedback program restricted to  buy-only. We observe that the linear feedback strategy $\myx^{LF}$ is very sensitive to $\alpha$. For $\alpha > 1$, LF generates oscillatory strategies, cf.~the right panel of Figure \ref{fig:feedbackCasestudies}. This occurs because the linear feedback underestimates the convex impact of trading on $(D_n)$, so that it first over-shoots relative to the target $D_{n+1}$, then under-shoots, etc. The unconstrained reference strategy actually tries to sell $u^\circ_2 < 0$ in the second step $n=2$, as well as in the $4^{th}$ step. Forcing $\myx^{LF}_n \ge 0$, we still obtain strong oscillations, no buying at $n=2$ and almost no buying at $n=3$. 
In contrast, the (estimated) optimal strategy smoothly maintains $\myx_n > 0$ throughout. 
Conversely for $\alpha < 1$ (left panel), the LF strategy underpurchases in the first trade since it does not correctly judge the reduced impact of large trades due to the concave price impact function. 

For $\alpha = 0.9$, price impact is much weaker and as a result the inventory urgency penalty dominates and leads to an L-shaped strategy, with a lot of buying in the first step or two, and a trickle for the rest of the steps. For $\alpha = 1.1$, price impact is the dominant feature and yields U-shaped strategies, with largest trades at $n=1$ and $n=N$. We also note that the dependence of the strategies on $\kappa,\eta$ remains complex for $\alpha \neq 1$. Since in all cases, total trades must add up to $X_0=100,000$, as parameters change, the resulting effect on $\myx_n$ is non-monotone, i.e., more is traded at some steps and less in others. As a result, the various (average) strategy curves cross each other, often more than once. 

Table \ref{tbl:parametric} reports the performance of an NN solver versus the LF one. Specifically, we use a 5D NN solver that is parametric in $(x,d,\kappa, \eta, \alpha)$.
The NN strategy is found to lead to total execution costs that are 10-15\% cheaper than LF, indicating that the differences observed in the respective strategies in Figure~\ref{fig:feedbackCasestudies} are material. Gains are larger for larger $\kappa$ and larger $\eta$, i.e., for configurations where price impact is more short-lived and more severe. While we do not have a ``gold standard'' to compare against, we may use the $\alpha=1.0$ performance (where LF is exact up to the non-negativity constraint that is almost never binding for this parameter configuration) as a yardstick for NN solution quality, since the NN solver is trained across different choices of $\alpha$.

\begin{table}[!htb]
\begin{center}
\begin{tabular}{crrrrc}
  & \multicolumn{2}{c}{$\kappa = 0.4$} & \multicolumn{2}{c}{$\kappa=0.8$} & \\ 
  \cline{2-3} \cline{4-5}
  & $\eta=1/3000$ & $\eta=1/500$ & $\eta=1/3000$ & $\eta=1/500$ & \\ \cline{2-5}
$\alpha=0.9$ & 10.33 & 7.33 & 15.05 & 10.58  \\
$\alpha=1.0$ & $-0.20$ & $-0.07$ &  $-0.52$ & $-0.04$  \\
$\alpha=1.1$ &  10.34 & 8.70 & 13.01 & 5.85 &   \\ \hline
\end{tabular}
 \caption{Performance of NN 5D solver in $(x,d,\kappa,\eta,\alpha)$ vs LF; based on $10^4$ simulations with $\sigma = 1$ and inventory penalty of $\nu=0.0001$. We report relative percent difference of total execution cost using Linear Feedback as baseline. Thus, positive values mean that the NN yields lower costs than LF, and negative means that LF outperforms. We expect positive values for $\alpha \neq 1$ and small negative values for $\alpha=1$. \label{tbl:parametric}}
\end{center}
\end{table}

\subsection{NN Implementations}\label{sec:nn-config}

NN solvers necessarily contain multiple tuning parameters that must be chosen during implementation. It is a folk theorem that some finetuning is always necessary, one of the reasons that  using NN is a bit of a ``black art''. 

In our setting, the feed-forward neural networks for $\hat{V}_n$ and $\hat{\myx}_n$ are very straightforward and do not require any bells and whistles. As mentioned, this simplicity of implementation is one of the reasons that we advocate this problem as a good pedagogical case study for building NN solvers for stochastic control problems. In particular, the NN architecture plays little significance: as long as one has sufficient flexibility, the choice of the number of neurons, the number of layers, etc., is very much secondary. Consequently, we default to the ``canonical'' setup with 16 neurons and 3 layers, that has been used in multiple prior works.

Nevertheless, there are certainly some tuning parameters that affect performance, first and foremost the effort spent on training. The latter is driven by the size of the training set, and the number of training epochs. The next set of experiments investigates in more detail how these parameters impact solution quality. In Table~\ref{tbl:linear-pnl} we compare the average P\&L of the NN strategy relative to LF for $\alpha=1.0$ as we vary the number of epochs $E$ and the number of training points $M$. As expected, accuracy increases as either $M$ or $E$ increase.
We observe that the running time is roughly linear in $M$ (since the latter is a straightforward loop) and sub-linear in $E$.

In addition, we also compare solvers that live in different dimensions, taking advantage of the fact that our 
implementation is fully dimension-agnostic and a single line change is needed to change $\ell$. For NN training purposes, that dimension $\ell$ does not matter, so the running time is constant as $\ell$ changes. However, as expected, smaller $\ell$ implies more dense training sets (since the volume of the training domain shrinks) and hence better accuracy. This reflects the fundamental property that learning a functional approximator is more laborious on a larger domain, and the respective ``volume'' grows in $\ell$. Thus, for the same $M,E$, a 3D solver will be more accurate than a 4D one, and less accurate than a 2D one. This is the (computational) price to pay for learning simultaneously across multiple parameters.

\begin{table}[!htb]
\centering 
\begin{tabular}{lcccr} \hline
& NN Configuration & $\kappa=0.4$ & $\kappa=0.8$ & Time (min)  \\ \hline\hline
4D w/$(x,d,\kappa,\eta)$ & $M=1000,E=1000$ & 0.10\% & 1.43\% & 2.73\\
& $M=2000,E=1000$ & 0.20\% & 0.88\% & 5.39\\
& $ M=2000, E=2000$ & 0.05\% & 0.46\% & 8.66 \\
& $M=4000, E=2000$ & 0.03\% & 0.28\% & 16.97  \\
& $M=8000, E=3000$ & 0.02\% & 0.27\% & 42.36 \\ \hline
3D w/$(x,d,\kappa)$ & $M=2000, E=2000$ & 0.04\% & 0.24\% & 8.38 \\
2D w/$(x,d)$ & $M=2000, E=2000$ & 0.03\% &  0.15\% &  8.14 \\ \hline
\end{tabular}
\caption{Average P\&L error compared to the linear feedback reference strategy for $\alpha=1$ across different NN implementations. In all cases, the NN has 3 layers with 16 neurons in each. $M$: number of training inputs; $E$: number of training epochs. The other test parameters are fixed at $\eta = 0.0001, \nu=0.0001, \sigma=1$, initial condition $X_0 = 100,000, D_0 = 0, N=10$, utilizing a 4D solver in $(x,d,\kappa, \eta)$, except for the last two rows. All running times are based on a Intel i7-11370H 3.0GHz laptop with 32GB RAM. \label{tbl:linear-pnl}}
\end{table}

Figure \ref{fig:comp_nn_strat} further compares solvers in different dimensions against each other in the nonlinear price impact case $\alpha = 1.1$. We consider NN solvers in 2D that only take $(x,d)$ coordinates, as well as in 3D with $(x,d,\kappa), (x,d,\eta)$ coordinates, in 4D with $(x,d,\kappa, \eta)$ and finally in 5D with $(x,d,\kappa,\eta,\alpha)$. The fact that all solvers yield very similar strategies is an empirical indication of the convergence of the NNs.

\begin{figure}[ht]
    \centering
    \begin{tabular}{cc}
    \includegraphics[height=1.8in,trim=0in 0in 0.2in 0in]{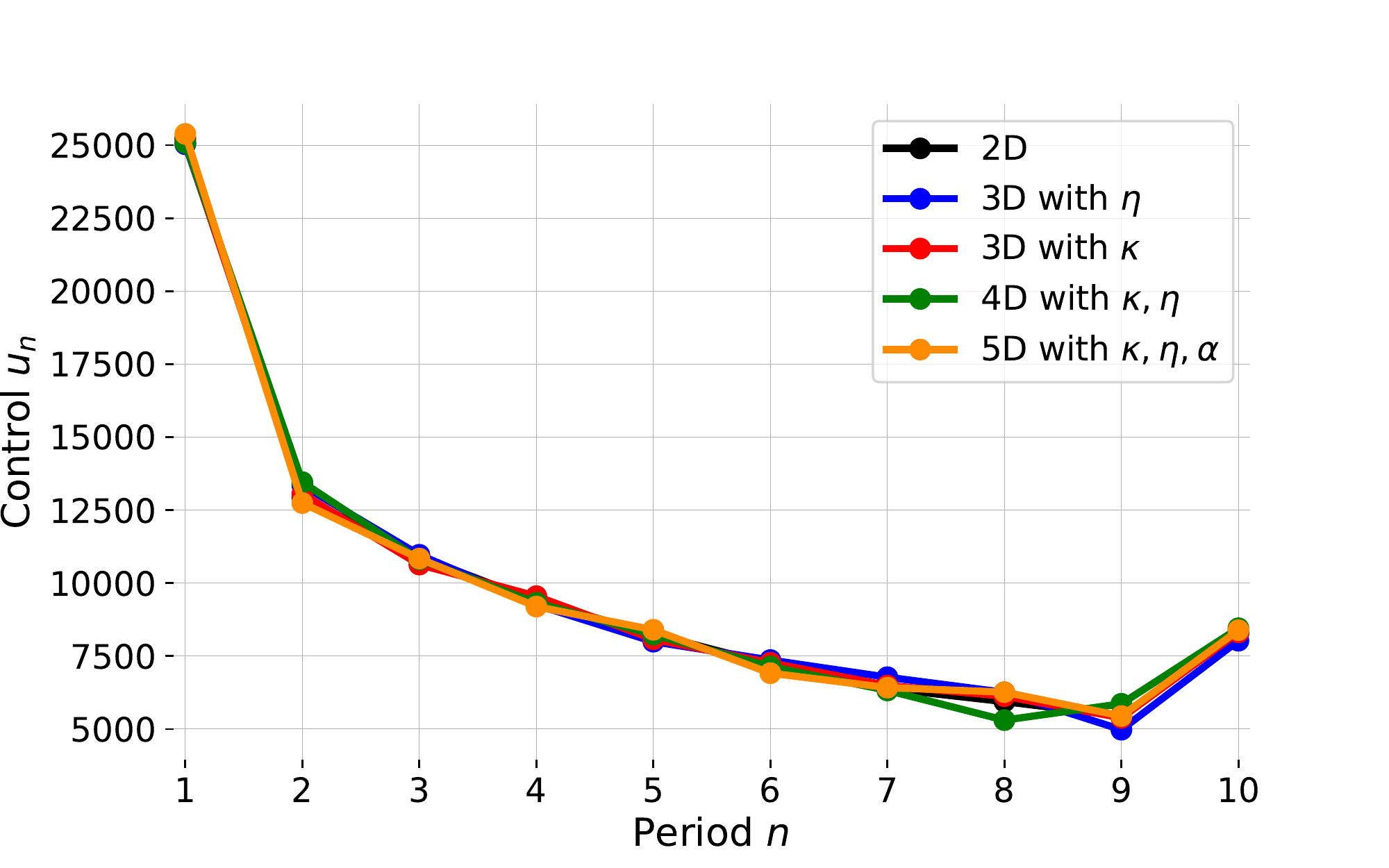} &
    \includegraphics[height=1.8in,trim=1in 0.4in 1in 0.6in,clip=true]{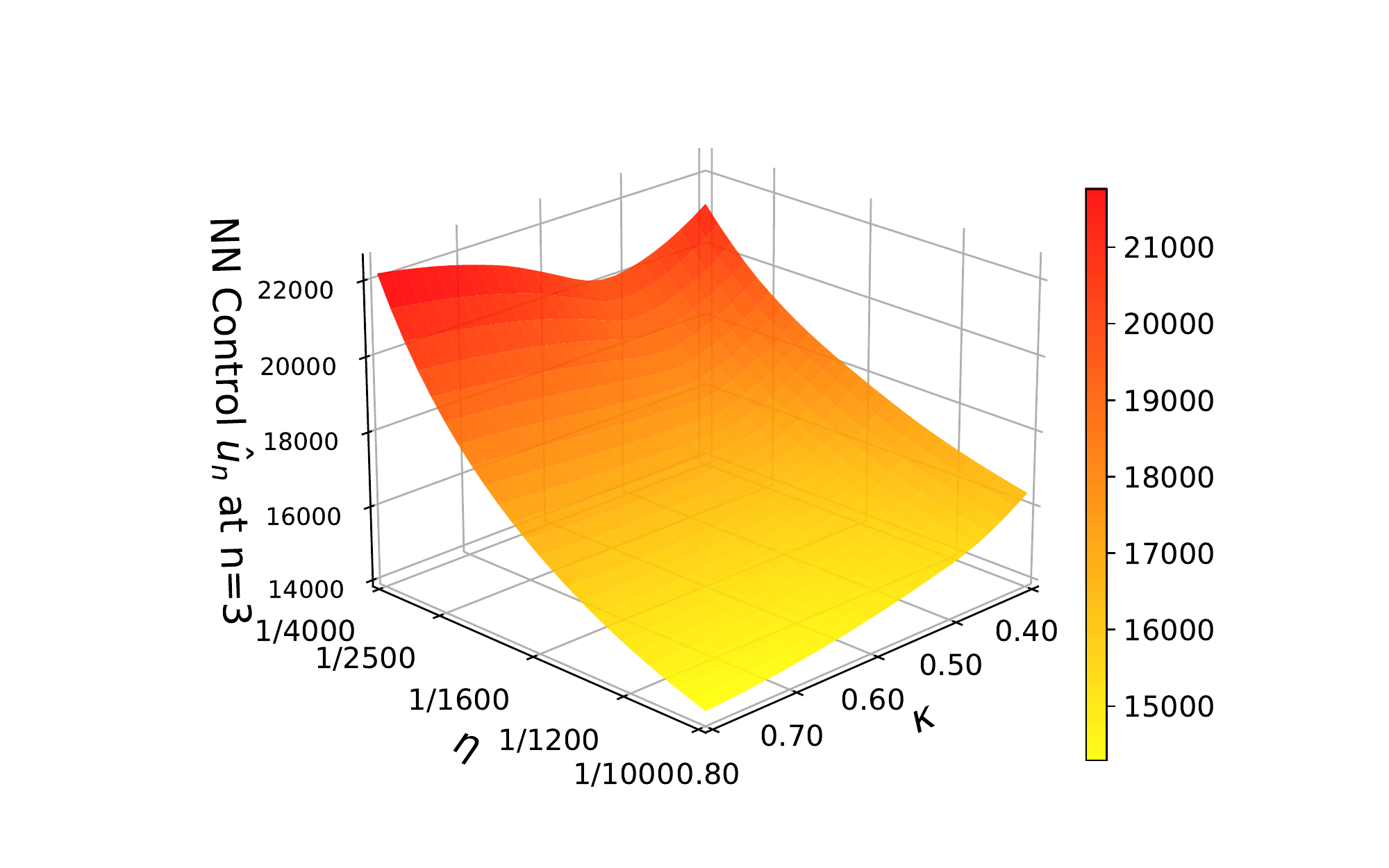}
    \end{tabular}
    \caption{\emph{Left:} Comparing different NN solvers with $\alpha=1.1$ and $\kappa=0.6, \eta=0.001, \nu=0.00005$. \emph{Right:} Dependence of NN-based execution $\myx_n(x,d, \kappa, \eta)$ strategy on $\kappa$ and $\eta$, keeping other parameters fixed at $n=3, x = 60,000, d=20$. We use the 5D solver shown on the left.}
    \label{fig:comp_nn_strat}
\end{figure}

Finally, the right panel of Figure~\ref{fig:comp_nn_strat} shows the fitted dependence of the control $\myx_n$ on $\kappa,\eta$ at $n=3$. Such dependence plots are the raison d'\^etre of parametric solvers, providing the modeler with a direct view of the sensitivity of the strategy to model parameters. Without a parametric solver, it would be prohibitively expensive to generate such surfaces through re-solving each configuration one-by-one. We observe that at this intermediate step $\myx_n$ shrinks in $\eta$ and in $\kappa$.

\subsection{Square-Root Price Impact} \label{sec:suqareroot}

We next investigate the special case where $\alpha = 0.5$. This ``square-root law" of price impact is advocated by some practitioners (c.f., \cite{LilloFarmerMantegna:03, AlmgrenThumHauptmannLi:05, BacryIugaLasnier:15}) and has also been addressed numerically in~\citet{CuratoGatheralLillo:16} via a brute force optimization of the cost function. In contrast, our case study highlights once more the usefulness of our parametric solver in that it readily computes an optimal execution schedule jointly across a range of different price impact parameters $\kappa$ and $\eta$, as well as urgency rates $\nu$; and thus unveils with ease after a single round of training the solution's dependence on the latter under this square-root price impact regime. Specifically, for $\alpha = 0.5$ fixed, we train a 5D NN solver that is parametric in $(x, d, \kappa, \eta, \nu)$ over the hyper-rectangle $x \in [0,10^5]$, $d \in [0,1]$, $\kappa \in [0.35, 0.85]$, $\eta \in [1/100, 1/1000]$, $\nu \in [0,10^{-5}]$ with $M=8000$ i.i.d.~uniformly selected training points, and set $\sigma = 0.1$.

\begin{figure}[ht] 
\centering
\begin{tabular}{cc}
\includegraphics[scale=.5]{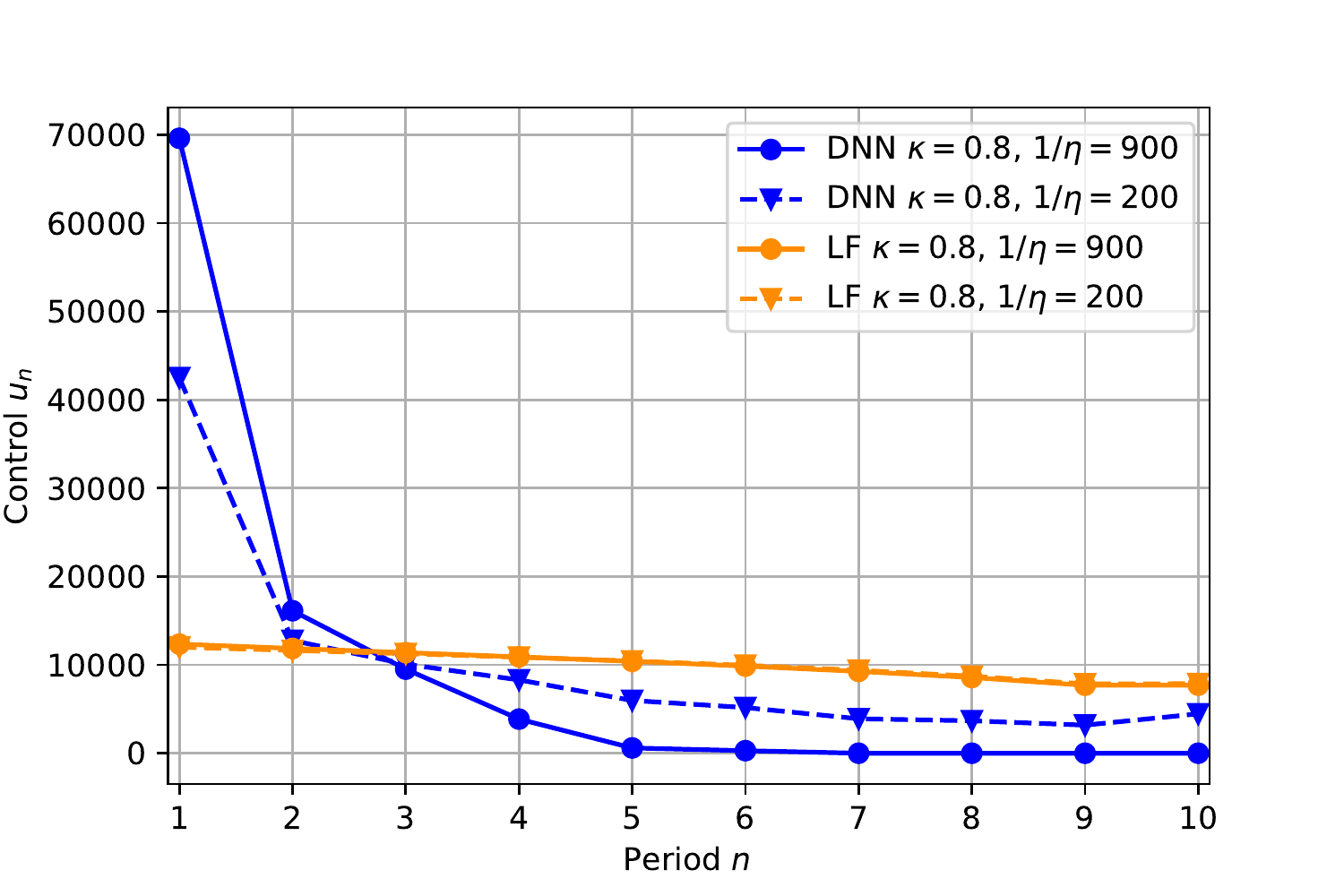} &
\includegraphics[scale=.5]{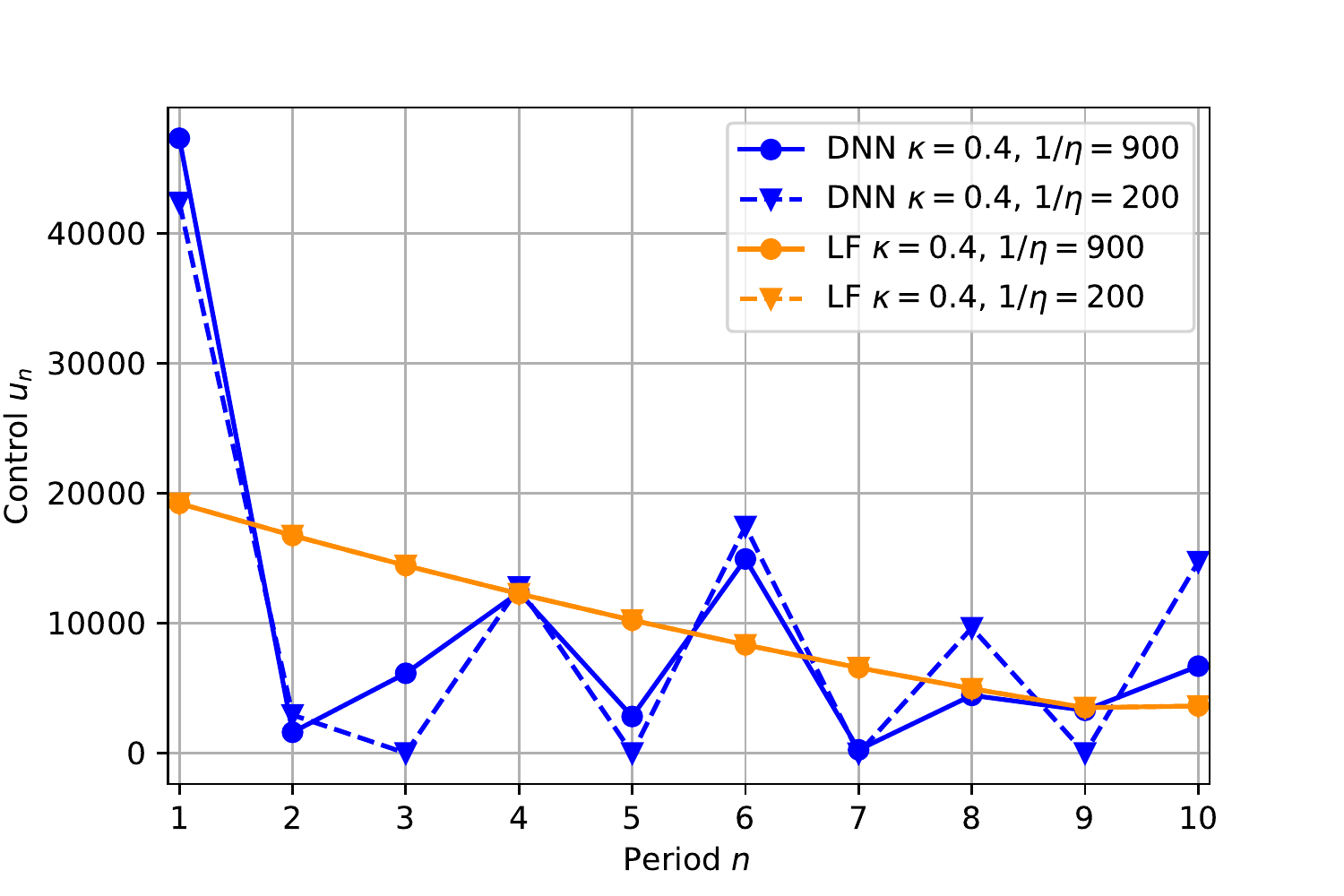} \\ 
$\nu=10^{-6}$ & $\nu = 0$
\end{tabular}
\caption{Execution strategy (mean values across $M'=10^4$ simulations) for $X_0 = 100,000, D_0=0$ and $N=10$ in the square root case $\alpha = 0.5$ with $\sigma = 0.1$, different price impact parameters $\kappa$ and $\eta$, and different urgency parameters $\nu$, utilizing a 5D solver in $(x,d,\kappa,\eta,\nu)$.
\label{fig:squarerootCasestudies}}
\end{figure}

Figure~\ref{fig:squarerootCasestudies} compares the resulting NN-based strategies to the linear feedback benchmark case $u^{LF}$ and Table~\ref{tbl:squareroot} reports their performances. We first note that the obtained strategy is very sensitive to whether the urgency parameter $\nu$ is zero or not. This is sensible because with $\alpha=0.5$, the magnitude of the incurred price impact is on the small scale and even further reduced by large values for $1/\eta$. As a consequence, as soon as $\nu$ becomes nonzero, the focus on rapidly reducing outstanding inventory dominates the overall order schedule; see the left panel in Figure~\ref{fig:squarerootCasestudies}. In contrast, when inventory control is turned off (i.e., $\nu = 0$) the NN strategy exhibits an oscillatory behavior: peaks of large buy orders are interrupted by near to zero-volume orders; see the right panel in Figure~\ref{fig:squarerootCasestudies}. This observation is somewhat consistent with the numerical results presented in~\cite[Section 4.4]{CuratoGatheralLillo:16} where the computed strategy consists of a few bursts of buying interspersed with long periods of no trading. In terms of total execution costs, the NN solver always significantly outperforms the LF benchmark strategy in all considered parameter configurations for $\kappa, \eta,\nu$; see Table~\ref{tbl:squareroot}. Apparently, the LF feedback policy overestimates the price impact which leads to an overall sub-optimal behavior.

\begin{table}[!htb]
\begin{center}
\begin{tabular}{llrr}
  & & $\eta=1/900$ & $\eta=1/200$ \\ \cline{3-4}
$\nu=10^{-6}$ & $\kappa = 0.8$ & 54.85 & 17.56  \\ 
$\nu=0$ & $\kappa = 0.4$ &  30.40 & 28.11
 \\ \hline
\end{tabular}
 \caption{Performance of the 5D NN solver with  $\my=(x,d,\kappa,\eta,\nu)$ vs.~LF; based on $10^4$ simulations with $\sigma = 0.1$. We report relative percent difference of total execution cost using Linear Feedback as baseline.  Positive values mean that the NN yields lower costs than LF. \label{tbl:squareroot}}
\end{center}
\end{table}

\subsection{Number of Periods}
Our NN algorithm trivially scales in the number of periods $N$ as it involves a simple loop over $n=N-1,\ldots,1$. Financially speaking, the most common interpretation is to fix the business time horizon $T$, and then pick the scheduling interval $\Delta t$, so that $N=T/\Delta t$. To illustrate this, we consider taking a larger number of intervals $N \gg 10$. To reflect the idea that $T$ is fixed, we need to re-scale some of the model parameters in terms of the frequency $\Delta t$. Specifically, since the dynamics of $(D_n)_{n=1,\ldots,N}$ are motivated by a discretization of a continuous-time kernel decay specification, we need to keep $\kappa^{(N)} \Delta t$ constant, in other words $\kappa^{(N)} \propto N^{-1}$, where we explicitly indicate the dependence of the resilience parameter on $N$. Indeed, as $N$ gets bigger, $\kappa^{(N)}$ should shrink so that the persistence of $D_n$ increases. Similarly, the continuous-time volatility $\sigma^{(N)} \sqrt{\Delta t}$ should be invariant, so that $\sigma^{(N)} \propto {N}^{-1/2}$. Finally, the inventory penalty is also proportional to business time, hence $\nu^{(N)} \propto N^{-1}$. In contrast, the impact parameter $\eta$ has no time-units, hence does not change as a function of $\Delta t$.

Returning to the closed-form formulas for $\alpha=1$, we observe that more fine scheduling frequency does not alter the fundamental U-shape (in the case $\nu=0$). In fact the initial and final trade amounts are only slightly affected by higher $N$, while the intermediate trades are roughly inversely proportional to $N^{-1}$, maintaining the same trading rate $u /\Delta t$ in business time.
 Similar intuition carries over to the solution when $\alpha > 1$. The left panel of Figure \ref{fig:largeN} shows the execution strategy for $N=30$ and $\alpha = 1.1$, keeping all other parameters as in Figure \ref{fig:feedbackCasestudies} right, modulo re-scaling as described above. The right panel of Figure \ref{fig:largeN} shows the execution strategy for $N=30$ and $\alpha=0.5$. The latter plot is comparable to the right panel of Figure \ref{fig:squarerootCasestudies}, taking zero inventory penalty $\nu=0$. However, to prevent negative deviations $D_n$, we must significantly decrease the noise amplitude to $\sigma=0.01$ (compared to $\sigma=0.1$ in Figure \ref{fig:squarerootCasestudies}).
 
For $\alpha = 1.1$ and $N=30$, the linear feedback strategy overshoots greatly, leading (for $\kappa=0.133$) to no trades for periods $n=2,\ldots,7$, while the NN solution maintains a steady pace (with a decreasing trend due to the inventory penalty $\nu$) throughout. It gains about 10\% in average cost savings compared to LF (11.2\% for $\kappa=0.4/3$ and 8.3\% for $\kappa=0.8/3$).

For $\alpha=0.5$ and $N=30$, we observe in Figure \ref{fig:largeN} bursty trading taking advantage of the concave price impact which discourages making small trades. Instead, the strategy executes every 2-4 periods, letting $(D_n)$ mean-revert back to zero in between. As in Figure \ref{fig:squarerootCasestudies}, due to $D_0 =0$ the first trade is very large, in our case about 65,000 for $\eta=1/900$ and about 78,000 for $\eta = 1/200$.
 
While the proposed dynamic programming approach implies error back-propagation as the iterations over $n$ are stepped through, we see very reasonable performance as $N$ is increased \emph{with exactly the same code}. We do recommend to employ a different strategy (e.g., RL-like with a single neural network that incorporates time-dependence) for $N \ge 50$ or so.

\begin{figure}[ht] 
\centering
\begin{tabular}{cc} 
\includegraphics[height=2in,trim=0.2in 0.2in 0.2in 0.5in]{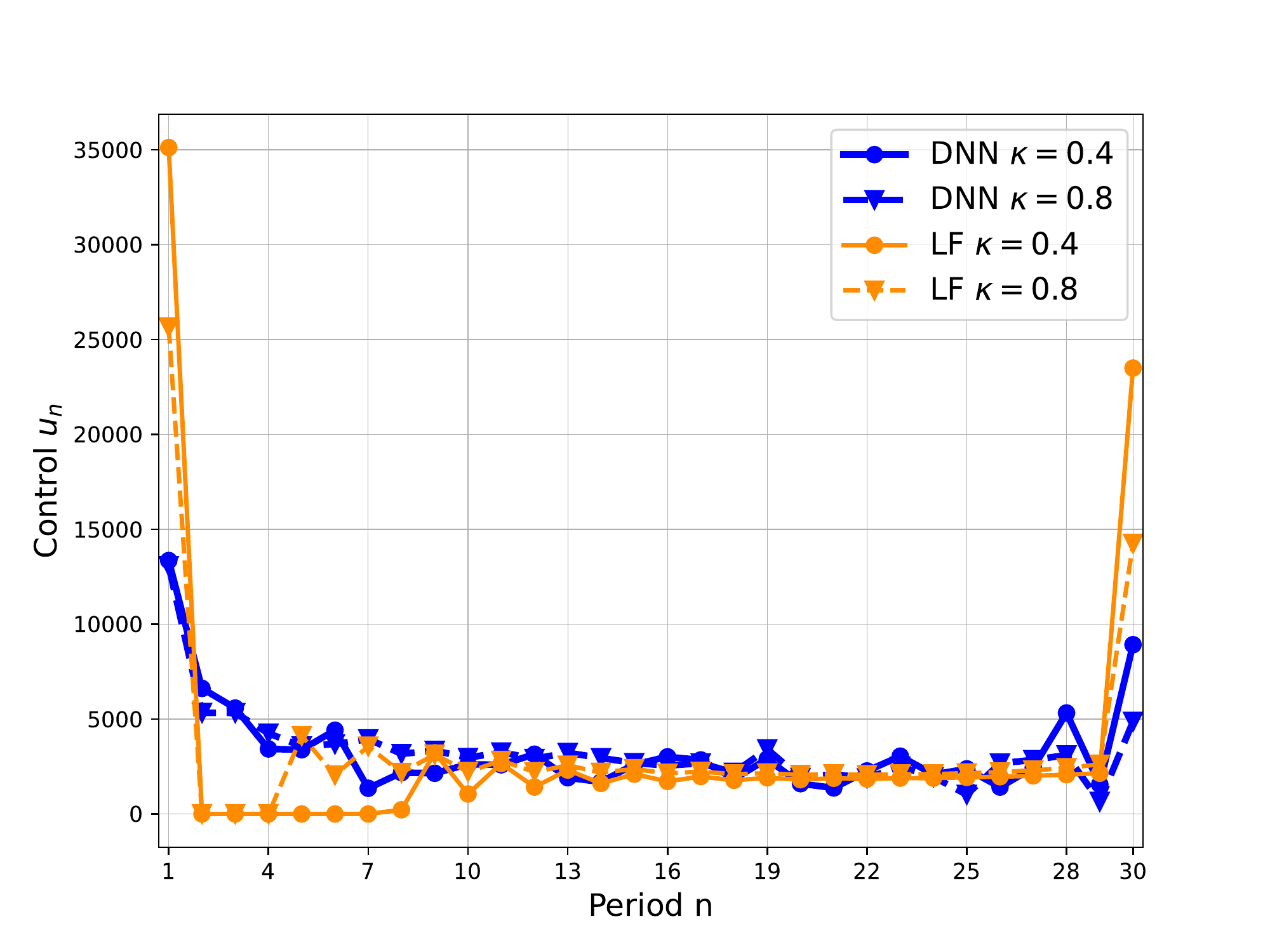} & 
\includegraphics[height=2in,trim=0.2in 0.2in 0.2in 0.5in]{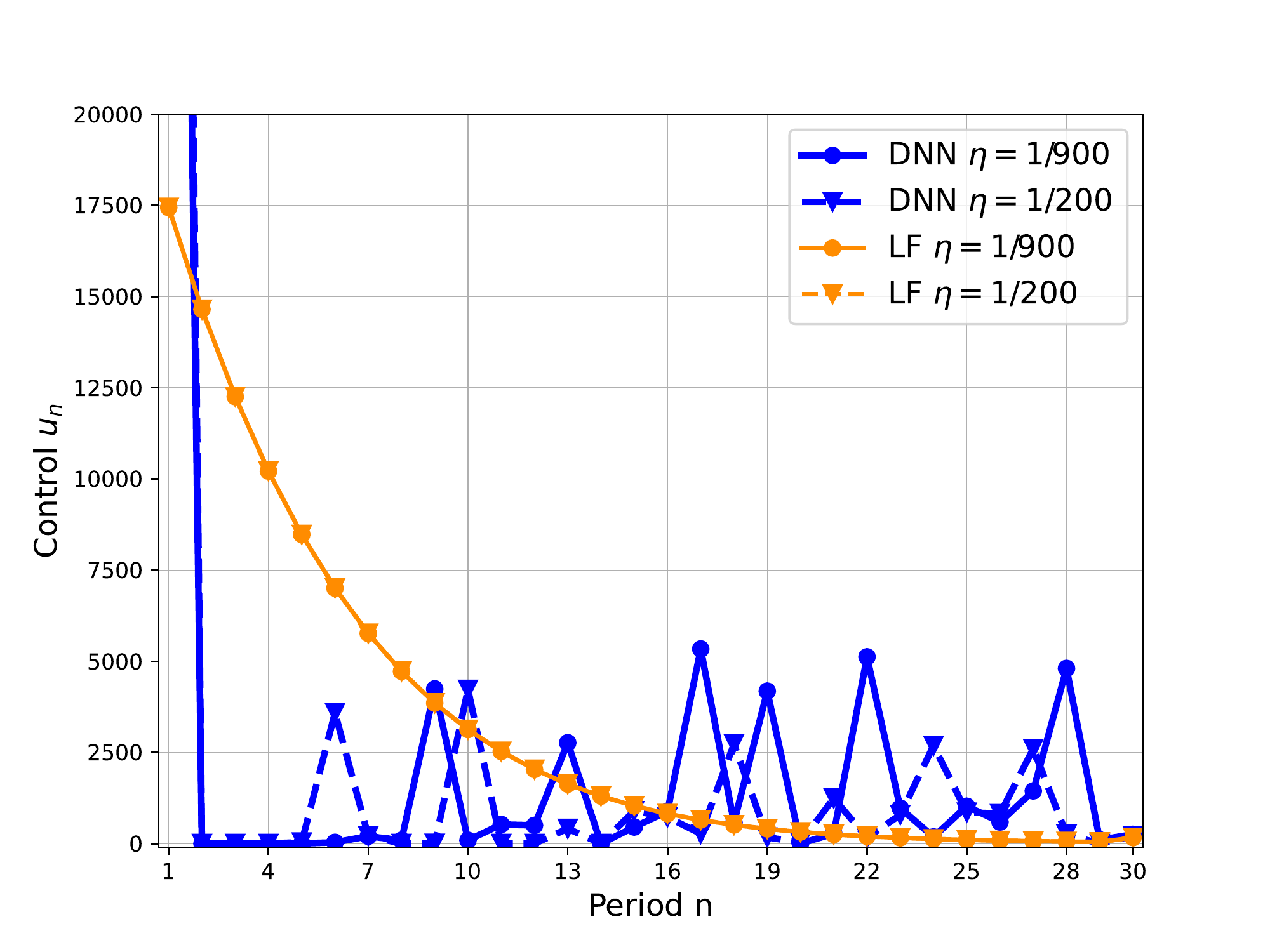}\\
$\alpha =1.1$ & $\alpha=0.5$
\end{tabular}
\caption{Execution strategy (mean values across $M'=10^4$ simulations) for $X_0 = 100,000,D_0=0$ and $N=30$ periods. Left panel: $\alpha=1.1$ with $\sigma = 1/\sqrt{3}$, urgency parameter $\nu = 0.0001/3$, $\eta=1/500$, utilizing a 5D solver in $(x,d,\kappa,\eta,\alpha)$. 
Right panel: square-root price impact $\alpha=0.5$ with $\sigma=0.01$, $\nu=0$ and $\eta=1/200$, utilizing a 4D solver in $(x,d,\kappa,\eta)$.
\label{fig:largeN}}
\end{figure}

\subsection{Multi-Exponential Decay Kernel}
We conclude our numerical experiments by investigating the performance of our NN algorithm on an extended version of our model formulated in Section~\ref{sec:model}. Recall that the transient price impact process introduced in~\eqref{def:deviation} for the studied buying program $u_n \geq 0$, $n=1,\ldots,N$, is given by
\begin{align*}
    D_{n} = & \, (1-\kappa)^n d_0 + \sum_{j=1}^{n} (1-\kappa)^{n-j} \eta \myx_j^\alpha + \sum_{j=1}^{n} (1-\kappa)^{n-j} \epsilon_j \qquad (n = 0, \ldots, N).
\end{align*} 
This deviation process belongs to a general class of decaying price impact processes, also called \emph{propagator models} originally developed by~\citet{BouchaudEtAl:09, BouchaudEtAl:04}, and can be written as
\begin{align} \label{eq:deviation_kernel}
    D_{n} = & \, G_{n,0} d_0 + \sum_{j=1}^{n} G_{n,j} \eta \myx_j^\alpha + \sum_{j=1}^{n} G_{n,j} \epsilon_j \qquad (n = 0, \ldots, N)
\end{align}
with an exponential decay kernel 
\begin{align} \label{eq:kernel}
G_{n,j} = (1-\kappa)^{n-j} \qquad (0 \leq j \leq n \leq N).
\end{align}
Therefore, it is very sensible to consider extensions of our model by allowing for different decay kernels $(G_{n,j})_{0 \leq j \leq n \leq N}$ in~\eqref{eq:kernel}. Specifically, there is empirical evidence reported in the literature that price impact exhibits some memory effect and rather decays according to a power law function; cf., e.g.,~\citet{BouchaudEtAl:04}.

In order to retain a Markovian framework, we study in this section the generalization where the kernel $G$ is a convex combination of exponentially decaying kernels, namely it is of the form
\begin{align} \label{def:multExpKernel}
G_{n,j} = \sum_{m=1}^M \zeta_m (1-\kappa_m)^{n-j} \qquad (0 \leq j \leq n \leq N)
\end{align}
for some $\zeta_m \in [0,1]$ with $\sum_{m=1}^M \zeta_m = 1$, and $\kappa_m \in (0,1]$. Plugging this back into~\eqref{eq:deviation_kernel} yields
\begin{align} \label{eq:deviation_mixed}
    D_{n} = & \, \sum_{m=1}^M \zeta_m D^m_n\qquad (n=0,\ldots, N),
\end{align}
where
\begin{equation*}
    D^m_n :=  (1-\kappa_m)^n d_0 + \sum_{j=1}^{n} (1-\kappa_m)^{n-j} \eta \myx_j^\alpha + \sum_{j=1}^{n} (1-\kappa_m)^{n-j} \epsilon_j 
  \qquad (n=0,\ldots, N).
\end{equation*}
In other words, the total price distortion $(D_n)_{n=0,\ldots,N}$ in~\eqref{eq:deviation_mixed} is now driven by $M$ processes $(D^1_n, \ldots, D^M_n)_{n=0,\ldots,N}$ which decay at  different timescales $\kappa_m$. The $D^m$'s are fully correlated with the dynamics $D^m_0 = d_0$ and
\begin{equation*}
    D^m_n =  \, (1-\kappa_m) D^m_{n-1} + \eta \myx_n^{\alpha} +
    \epsilon_n. 
\end{equation*}

Moreover, the dynamics in~\eqref{def:deviation} generalize to $D_0 = d_0$,
\begin{equation*}
    D_{n} =   \sum_{m=1}^M \zeta_m (1-\kappa_m) D^m_{n-1}  + \eta \myx_n^{\alpha} + \epsilon_n \qquad (n=1,\ldots, N).
\end{equation*}
In the special case $\zeta_m = 1$, $\zeta_n=0$ for all $n\neq m$ we retrieve the original setup from~\eqref{def:deviation}.

Augmenting the $d^m$'s to the state space and 
following the same reasoning as in Section~\ref{subsec:execution}, the objective function in~\eqref{def:optProblem} becomes
\begin{equation*} 
\inf_{(\myx_n)_{n=1,\ldots,N}\in\mathcal{A}_1}
  \mathbb{E}\left[ \sum_{n=1}^N \left\{\left( \sum_{m=1}^M \zeta_m (1-\kappa_m) D^m_{n-1} +\frac{\eta}{2} \myx_{n}^{\alpha} \right) 
    \myx_{n} + \mynu (X_{n-1} - \myx_n)^2 \right\}\right];
\end{equation*}
the corresponding value functions in~\eqref{def:valuefunction} are given by
\begin{equation*}
   \begin{aligned}
   & V_n(x, d^1, \ldots, d^M) := \inf_{(\myx_j)_{j=n,\ldots,N} \in \mathcal{A}_{n} } \mathbb{E}\Bigg[ \sum_{j=n}^N \bigg\{ \left( \sum_{m=1}^M \zeta_m (1-\kappa_m) D^m_{n-1} +\frac{\eta}{2} \myx_{j}^{\alpha}\right) \bigg.
   \myx_{j} \Bigg. \\
   & \Bigg. \hspace{160pt} \bigg. + \mynu (X_{j-1} - \myx_j)^2 \bigg\} \, \bigg\vert \, X_{n-1}=x, D^m_{n-1}=d^m, m=1,\ldots,M \Bigg]
   \end{aligned}
\end{equation*}
for all $n \in \{1,\ldots,N\}$.

We next illustrate the above extension in the case $M=2$ so that the decay kernel is a mixture of two exponentials. Taking $\zeta_1 \equiv \zeta, \zeta_2 \equiv 1-\zeta$ for the mixing weight $\zeta \in [0,1]$, the associated dynamic programming (DP) equation in~\eqref{eq:bellman-1} and~\eqref{eq:bellman-2} modifies to
\begin{align} \label{eq:bellman-1.2}
V_{N}(x, d^1, d^2) 
    & = \big( \zeta (1-\kappa_1) \cdot d^1 + (1-\zeta) (1-\kappa_2) \cdot d^2 \big) \cdot x + \frac{\eta}{2} x^{\alpha+1} \qquad \text{and } \\
 \label{eq:bellman-2.2}
V_{n}(x, d^1, d^2) 
    & = \inf_{\myx \in [0,x]} \mathbb{E}\Bigg[ \Big( \zeta (1-\kappa_1) \cdot d^1 + (1-\zeta) (1-\kappa_2) \cdot d^2 \Big) \cdot \myx + \frac{\eta}{2} \myx^{\alpha+1}  + \mynu ( x - \myx)^2  \Bigg. \nonumber \\
    & \hspace{52pt} \Bigg. + V_{n+1}(x-\myx,D^1_{n},D^2_{n}) \, \bigg\vert \, X_{n-1}=x, D^1_{n-1}=d^1, D^2_{n-1}=d^2\Bigg]
\end{align} 
for all $n=N-1,\ldots,1$.

Our NN agorithm can be easily extended to the DP equations in~\eqref{eq:bellman-1.2} and~\eqref{eq:bellman-2.2}. 
In particular, we can also treat the mixing weight $\zeta$ as a model hyperparameter that can be trained upon. In Figure~\ref{fig:mixedExp} we illustrate the execution profiles of the NN solver $(\hat{u}_n(x,d^1,d^2,\zeta))_{n=1,\ldots,N}$  for a model with a multi-exponential decay kernel as specified in~\eqref{def:multExpKernel} where $M=2$, $\kappa_1=0.4$, $\kappa_2=0.8$ and several values of  $\zeta$. We train the solver across the 4 inputs $(x,d^1, d^2,\zeta)$ with the latter in the range $\zeta \in [0.25, 0.75]$. 
The plot extends the configuration of the left panel in Figure~\ref{fig:feedbackCasestudies} to a multi-exponential propagator and shows the corresponding NN policy (mean values across $M'=10^4$ simulations) in the concave price impact regime $\alpha = 0.9$ with $\sigma=1$, $\nu=0.0001$ and $\eta=1/500$. As before, we take $X_0 = 100,000, D_0=0$ and $N=10$ periods and three different $\zeta$'s.

As the accompanying table shows, misspecifying the decay kernel is costly: the DNN solver beats an LF strategy that assumes linear price impact $\alpha=1$ and a single decay parameter $\kappa$ by 9-12\%. Of note, taking an exponential kernel with the larger $\kappa_2$ does better than using an averaged $\kappa= \zeta \kappa_1 + (1-\zeta)\kappa_2$. Even when the impact is linear, $\alpha=1.0$, the DNN beats the LF strategy (by 0.32\%-1.51\%) due to its mis-specification via an exponential kernel.

\begin{figure}[ht] 
\centering
\begin{tabular}{cc} 
\begin{minipage}{0.41\textwidth}
\includegraphics[width=\textwidth,trim=0.4in 0.2in 0.2in 0.5in]{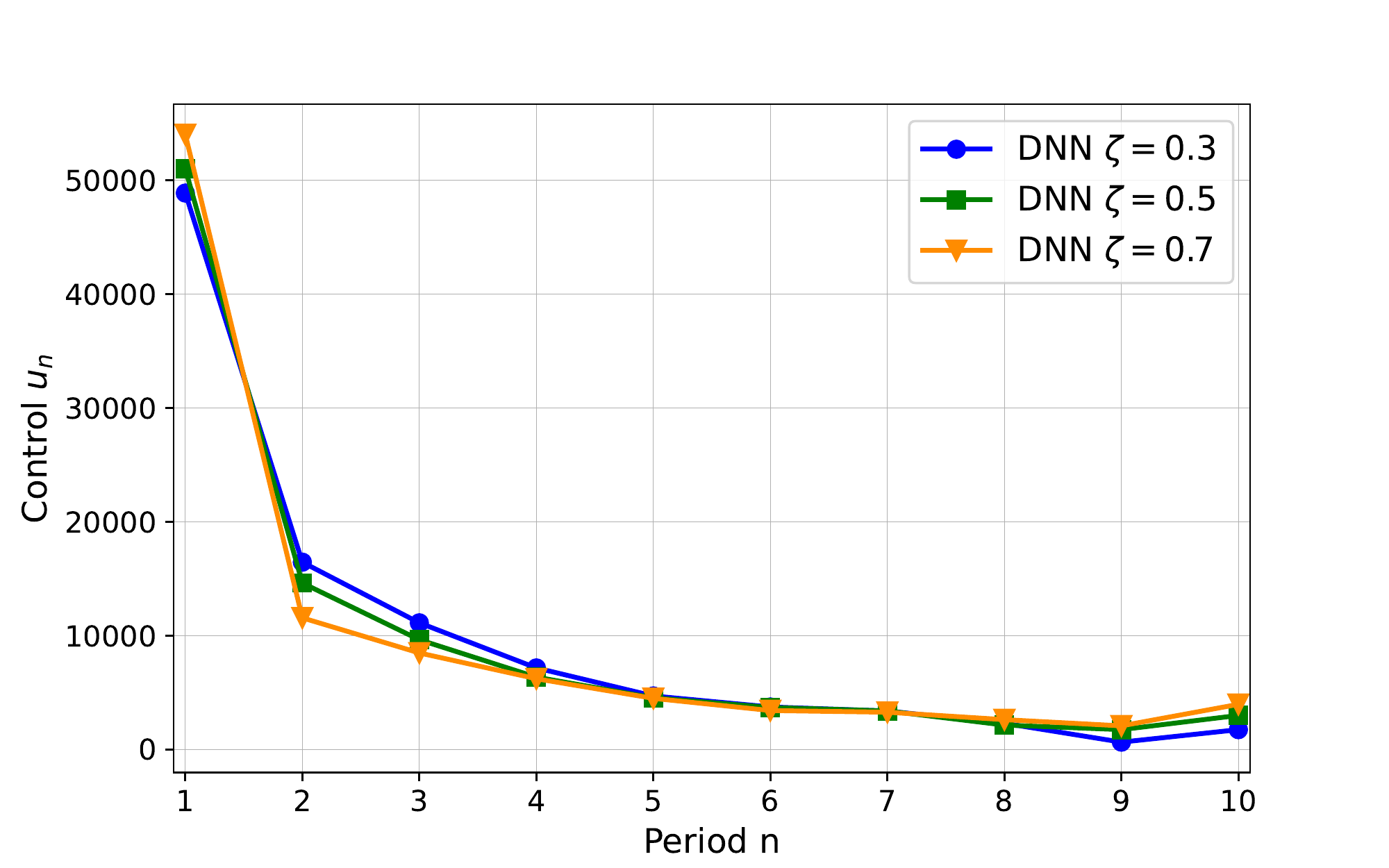}
\end{minipage} & \begin{minipage}{0.49\textwidth} 
\begin{tabular}{lrrr}
& \multicolumn{3}{c}{DNN w/$\alpha = 0.9$}  \\
 LF w/ &  $\zeta=0.3$ & $\zeta=0.5$ & $\zeta=0.7$ \\ \cline{2-4}
  $\kappa = \kappa_1 = 0.4$ & 4.81 & 12.12 & 11.03  \\ 
  $\kappa = \zeta \kappa_1 + (1-\zeta)\kappa_2$ & 5.18 & 12.37 & 10.02 \\
 $\kappa = \kappa_2 = 0.8$ &  5.92 & 12.98 & 9.02
 \\ \hline
  & \multicolumn{3}{c}{DNN w/$\alpha = 1$}  \\ \cline{2-4}
  $\kappa = \kappa_1 = 0.4$ &  & 1.51 &   \\ 
  $\kappa = 0.6$ &  & 0.32 &  \\
 $\kappa = \kappa_2 = 0.8$ &   & 0.44 & 
 \\ \hline
\end{tabular}
\end{minipage}

\end{tabular}
\caption{\emph{Left}: Execution strategy (mean values across $M'=10^4$ simulations) for $X_0 = 100,000,D_0=0$ and $N=10$ periods with bi-exponential decay kernel with $\kappa_1=0.4$, $\kappa_2=0.8$ and weight $\zeta \in \{0.3, 0.5, 0.7\}$.
We use $\alpha=0.9$ with $\sigma=1$, $\nu=0.0001$ and $\eta=1/500$, utilizing a 4D solver in $(x,d^1,d^2,\zeta)$. \emph{Right}: Percent gain of the DNN value function relative to the indicated LF comparator. Positive values mean that the DNN achieves lower execution costs.
\label{fig:mixedExp}}

\end{figure}

\section{Conclusion}\label{sec:conclude}

In this article we have investigated neural network surrogates for solving optimal execution problems across a range of model parameters. Our approach jointly learns an optimal strategy as a function of the stochastic system state and of the market configuration.

The developed algorithm and the accompanying Jupyter Notebook can be used as a starting point for many other related analyses. For example, it would be straightforward to modify the code to handle parametric stochastic control problems of similar flavor (e.g., discrete-time hedging). 

A further use case is to use the trained neural network as a building block in a more sophisticated setup. In particular, one may consider frameworks that explicitly account for model risk, in the sense of imprecisely known parameters. In the adaptive approach (including the Predictive Model Control popular in engineering), the modeler first develops \emph{learning dynamics}, which convert static parameters, such as $\kappa$, into  a stochastic process $(\hat{\kappa}_n)$, where $\hat{\kappa}_n$ is the best estimate of the book resilience at step $n$. Updating equations for $\hat{\kappa}_{n+1}$ in terms of the previous $\hat{\kappa}_n$ and new information from step $n+1$ (such as using the Bayesian paradigm) yield dynamics that can be merged with those of $(X_n, D_n)$. One then \emph{plugs-in} the resulting $\hat{\kappa}_n$ (and other similarly learned/updated parameters) into the NN-learned $\hat{u}(X_n, D_n, \hat{\kappa}_n)$ to obtain the adaptive strategy---which takes into account the latest parameter estimates, but does not solve the full Bellman equation. Conversely, one could also consider \emph{robust} approaches that {minimize} $\hat{V}_n(\cdot)$ over feasible parameter settings in order to protect against a worst-case situation. The latter again requires access to the computed $\hat{V}_n$ as a building block. Finally, we may mention the \emph{adaptive robust} approach \cite{Chen:18} that combines dynamic learning with a worst-case min-max optimization to protect against incorrect estimates or mis-specified dynamics. The resulting numerical algorithms will be investigated in a separate, forthcoming sequel.

\section{Proofs} \label{sec:proofs}

We start with Lemma~\ref{lem:positive} which provides an intermediate computation relevant for showing $V^\circ_n(0,0) < 0$ in~\eqref{prop:benchmark:valuefunction} (existence of round-trips) and characterizing the optimal $u^\circ$ in in~\eqref{prop:benchmark:feedback}. 

\begin{lemma} \label{lem:positive}
For all $n = N, N-1, \ldots, 2$ the constants $\mya_{n}, \myb_{n}, \myc_{n}$ recursively defined in~\eqref{prop:benchmark:recursionTerminal} and~\eqref{prop:benchmark:recursion} satisfy $2 \eta + 4 \mynu + 4 \mya_{n} - 4 \eta \myb_{n} + 4 \eta^2 \myc_{n} > 0$.
\end{lemma}

\begin{proof}
For $n=N$ we directly get from~\eqref{prop:benchmark:recursionTerminal} that 
\begin{equation*}
    2 \eta + 4 \mynu + 4 \mya_{N} - 4 \eta \myb_{N} + 4 \eta^2 \myc_{N} = 4\eta\kappa+4\nu > 0. 
\end{equation*}
For $n=N-1$, using the definition in~\eqref{prop:benchmark:recursion}, one computes
\begin{align*} 
    2 \eta + 4 \mynu + 4 \mya_{N-1} - 4 \eta \myb_{N-1} + 4 \eta^2 \myc_{N-1} = \frac{32\eta\kappa\nu + 16 \nu^2 + 4\eta^2\kappa^2 (4-\kappa^2)}{4\eta\kappa+4\nu} > 0.
\end{align*}
The general claim can then be checked similarly with a tedious backward induction relying on the following recursive relation obtained from~\eqref{prop:benchmark:recursion}
 \begin{align} 
     2 \eta + 4 \mynu + 4 \mya_{n} - 4 \eta \myb_{n} + 4 \eta^2 \myc_{n} = & \, 2 \eta + 4 \mynu + 4 (\mya_{n+1} + \mynu) - 4\eta(1-\kappa) \myb_{n+1} + 4 \eta^2 (1-\kappa)^2 \myc_{n+1} \nonumber \\
     & \, - \frac{\big( 2(2\mynu+2\mya_{n+1}-\eta \myb_{n+1})+2\eta(1-\kappa)(1-\myb_{n+1}+2\eta\myc_{n+1}) \big)^2}{2 \eta + 4 \mynu + 4 \mya_{n+1} - 4 \eta \myb_{n+1} + 4 \eta^2 \myc_{n+1}} \nonumber \\
    = & \, 4 \eta \kappa + 4 \mynu + 4 \eta^2 \kappa^2 \frac{2 \myc_{n+1} (2 \mya_{n+1}+2\nu-\eta)-(1-\myb_{n+1})^2}{2 \eta + 4 \mynu + 4 \mya_{n+1} - 4 \eta \myb_{n+1} + 4 \eta^2 \myc_{n+1}}. \label{proof:lem:eq1}
 \end{align}
\end{proof}

\textbf{Proof of Proposition~\ref{prop:benchmark}:}
In the linear case $\alpha=1$ the unconstrained version of the optimal execution problem formulated in~\eqref{def:valuefunctionUC} is a linear quadratic stochastic control problem. Therefore, it is well known that for all $n \in \{1,\ldots,N\}$ the value functions $V^\circ_n(x,d)$ are linear quadratic in $x$ and~$d$. This motivates the ansatz $V^\circ_n(x,d) = \mya_n x^2 + \myb_n x d + \myc_n d^2 + \mye_n$, where the coefficients $\mya_n, \myb_n, \myc_n, \mye_n \in \mathbb{R}$ are determined via backward induction by using the corresponding dynamic programming equations in~\eqref{eq:bellman-1} and~\eqref{eq:bellman-2}. 

First, the terminal condition in~\eqref{eq:bellman-1} yields $\mya_N = \frac{\eta}{2}, \myb_N = 1 - \kappa, \myc_N = 0$ as claimed in~\eqref{prop:benchmark:recursionTerminal}, as well as $\mye_N = 0$. Next, for the inductive step, let $n \in \{N-1, \ldots, 1\}$. The dynamic programming equation in~\eqref{eq:bellman-2} (for the considered unconstrained version of the problem) yields   
\begin{align}
V_{n}^\circ(x, d) 
    & = \min_{\myx \in \mathbb{R}} \bigg\{ (1-\kappa) d \myx + \frac{\eta}{2} \myx^{2} + \mynu ( x - \myx)^2 + \mathbb{E}\Big[ V_{n+1}^\circ(X_n,D_n) \, \Big\vert \, X_{n-1}=x, D_{n-1}=d \Big] \bigg\} \nonumber \\
    & = \min_{\myx \in \mathbb{R}} \bigg\{ (1-\kappa) d \myx + \frac{\eta}{2} \myx^{2} + \mynu ( x - \myx)^2 \nonumber \\
    & \hspace{42pt} + \mathbb{E}\Big[ V^\circ_{n+1} \big( x - \myx,(1-\kappa)d + \eta \myx + \epsilon_{n} \big) \, \Big\vert \, X_{n-1}=x, D_{n-1}=d \Big] \bigg\}. \label{proof:prop_bench:inductivestep}
\end{align}
Plugging in $V^\circ_{n+1}(x,d) = \mya_{n+1} x^2 + \myb_{n+1} x d + \myc_{n+1} d^2 + \mye_{n+1}$ in~\eqref{proof:prop_bench:inductivestep} and solving the squares we obtain
\begin{align}
    & V_{n}^\circ(x,d) \nonumber \\
    & = \min_{\myx \in \mathbb{R}}  \Bigg\{ \myx^2 \left( \frac{\eta}{2} +  \mynu + \mya_{n+1} - \eta \myb_{n+1} + \eta^2 \myc_{n+1}  \right) \nonumber \\
    & \hspace{45pt} + \myx \Big( ( \eta \myb_{n+1} -2  \mynu - 2 \mya_{n+1}) x + (1-\myb_{n+1} + 2\eta\myc_{n+1}) (1-\kappa) d \Big) \nonumber  \\
    & \hspace{45pt} + ( \mynu + \mya_{n+1}) x^2 + (1-\kappa) \myb_{n+1} x d + (1-\kappa)^2 \myc_{n+1} d^2 + \mye_{n+1} \nonumber \\
    & \hspace{45pt} + \left( \frac{1}{2} u \left( 1 - 2 \myb_{n+1} + 4 \eta\myc_{n+1} \right)  + (\myb_{n+1} x + 2 (1-\kappa) \myc_{n+1} d) \right) \mathbb{E} \left[  \epsilon_{n} \, \big\vert \, X_{n-1} = x, D_{n-1} = d \right]  \nonumber \\
    & \hspace{45pt} + c_{n+1} \mathbb{E} \left[  \epsilon^2_{n} \, \big\vert \, X_{n-1} = x, D_{n-1} = d \right] \Bigg\} \nonumber \\ 
    & = \min_{\myx \in \mathbb{R}}  \Bigg\{ \myx^2 \left( \frac{\eta}{2} +  \mynu + \mya_{n+1} - \eta \myb_{n+1} + \eta^2 \myc_{n+1}  \right) \nonumber \\
    & \hspace{45pt} + \myx \Big( ( \eta \myb_{n+1} -2  \mynu - 2 \mya_{n+1}) x + (1-\myb_{n+1} + 2\eta\myc_{n+1}) (1-\kappa) d \Big) \label{proof:prop_bench:min}  \\
    & \hspace{45pt} + ( \mynu + \mya_{n+1}) x^2 + (1-\kappa) \myb_{n+1} x d + (1-\kappa)^2 \myc_{n+1} d^2 + \mye_{n+1} + \myc_{n+1} \sigma^2 \Bigg\}, \nonumber
\end{align}
where we used the fact that $\mathbb{E} [ \epsilon_{n} \, \vert \, X_{n-1}, D_{n-1}] = \mathbb{E}[ \epsilon_{n}] = 0$ as well as $\mathbb{E} [ \epsilon_{n}^2 \, \vert \, X_{n-1}, D_{n-1}] = \mathbb{E}[ \epsilon^2_{n}] = \sigma^2$. Minimizing~\eqref{proof:prop_bench:min} with respect to $\myx$ gives
\begin{equation}
    \myx = - \frac{ (\eta \myb_{n+1} - 2  \mynu - 2 \mya_{n+1}) x + (1 - \myb_{n+1} + 2 \eta\myc_{n+1}) (1-\kappa) d}{\eta +  2 \mynu + 2 \mya_{n+1} - 2 \eta \myb_{n+1} + 2 \eta^2 \myc_{n+1}} \label{proof:prop_bench:feedback}
\end{equation}
and hence the feedback policy $u^\circ_n$ as claimed in~\eqref{prop:benchmark:feedback}. In particular, note that it follows from Lemma~\ref{lem:positive} that $\eta +  2 \mynu + 2 \mya_{n+1} - 2 \eta \myb_{n+1} + 2 \eta^2 \myc_{n+1} > 0$ and that $u$ in~\eqref{proof:prop_bench:feedback} is indeed the unique minimum in~\eqref{proof:prop_bench:min}. Moreover, inserting~\eqref{proof:prop_bench:feedback} back into~\eqref{proof:prop_bench:min} yields
\begin{align*}
    V^\circ_{n}(x, d) & = -\frac{\Big( ( \eta \myb_{n+1} - 2\mynu - 2 \mya_{n+1}) x + (1-\myb_{n+1} + 2\eta\myc_{n+1}) (1-\kappa) d \Big)^2}{2 \eta + 4 \mynu + 4 \mya_{n+1} - 4 \eta \myb_{n+1} + 4 \eta^2 \myc_{n+1}}  \nonumber\\
    & \hspace{13pt} + ( \mynu + \mya_{n+1}) x^2 + (1-\kappa) \myb_{n+1} x d + (1-\kappa)^2 \myc_{n+1} d^2 \nonumber \\
    & \hspace{13pt} + \myc_{n+1} \sigma^2 + \mye_{n+1} \nonumber \\
    & = \left( \mynu + \mya_{n+1} - \frac{(\eta\myb_{n+1} - 2 \mynu - 2 \mya_{n+1})^2}{2 \eta + 4 \mynu + 4 \mya_{n+1} - 4 \eta \myb_{n+1} + 4 \eta^2 \myc_{n+1}} \right) x^2 \nonumber \\
    & \hspace{13pt} + (1-\kappa) \left(\myb_{n+1} -\frac{(\eta\myb_{n+1}-2 \mynu - 2 \mya_{n+1}) (1-\myb_{n+1} + 2\eta\myc_{n+1})}{\eta + 2 \mynu + 2 \mya_{n+1} - 2 \eta \myb_{n+1} + 2 \eta^2 \myc_{n+1}} \right) x d  \nonumber \\
    & \hspace{13pt} + (1-\kappa)^2 \left( \myc_{n+1} -\frac{(1-\myb_{n+1}+2\eta\myc_{n+1})^2}{2 \eta + 4 \mynu + 4 \mya_{n+1} - 4 \eta \myb_{n+1} + 4 \eta^2 \myc_{n+1}} \right) d^2  + \myc_{n+1} \sigma^2 + \mye_{n+1},
\end{align*}
which implies the desired recursive formulas provided in~\eqref{prop:benchmark:recursion} and the representation of the value function in~\eqref{prop:benchmark:valuefunction}. Also note in~\eqref{prop:benchmark:feedback} that $u^\circ_1$ is just a deterministic constant in $\mathbb{R}$ and that $u^\circ_n$ is normally distributed for all $n=2,\ldots,N$. Indeed, since $u^\circ_n$ is linear in $X_{n-1}^\circ$ and $D_{n-1}^\circ$, and the state variables $X_{n-1}^\circ$ and $D_{n-1}^\circ$ are linear in $(\myx^\circ_j)_{j=1,\ldots,n-1}$ and the i.i.d.~zero-mean Gaussian noise $(\epsilon_j)_{j=1,\ldots,n-1}$, one checks that $u^\circ_n$ is ultimately just a linear transformation of $(\epsilon_j)_{j=1,\ldots,n-1}$. As a direct consequence, we can conclude that $(u^\circ_n)_{n=1,\ldots,N}$ is an admissible strategy in the set $\mathcal{A}^\circ_1$ as defined in~\eqref{def:admissibleSetUC}. Finally, the representation of $D^\circ$ in~\eqref{prop:benchmark:deviation} follows directly from its state dynamics in~\eqref{def:deviation}.  
\qed

\textbf{Proof of Proposition~\ref{prop:benchmark_det}:} In the case $\sigma = 0$ the optimal control problem in~\eqref{def:valuefunctionUC} (with $\alpha = 1$) is deterministic. We can introduce the corresponding cost functional $C: \mathbb{R}^N \rightarrow \mathbb{R}$ given by
\begin{align}
C(\myx_1,\ldots,\myx_N) & \triangleq
\sum_{n=1}^N \left\{ \left( (1-\kappa) D_{n-1} + \frac{\eta}{2} \myx_{n} \right) \myx_{n} + \mynu (X_{n-1} - \myx_n)^2 \right\} \nonumber \\
& = \frac{\eta}{2} \sum_{n=1}^N \myx_n^2 + (1-\kappa) \sum_{n=1}^N D_{n-1} \myx_n + \nu \sum_{n=1}^N (X_{n-1} - \myx_n)^2. \label{proof:benchmark_det:cost}
\end{align}
For all $n=1,\ldots,N$, considering the state variables $X_n(\myx_1,\ldots,\myx_N) \triangleq X_n$ and $D_n(\myx_1,\ldots,\myx_N) \triangleq D_n$ in~\eqref{def:inventory} and~\eqref{def:deviation} as functions in $\myx_1,\ldots,\myx_N \in \mathbb{R}$, we note that 
\begin{equation*}
    \frac{\partial X_n}{\partial \myx_i}  = -1, \quad \frac{\partial D_n}{\partial \myx_i} = \eta (1-\kappa)^{n-i} \qquad (1\leq i \leq n).
\end{equation*}
In particular, we have the relation
\begin{equation} \label{proof:benchmark_det:pD}
    \frac{\partial D_n}{\partial \myx_{i}} = (1-\kappa) \frac{\partial D_n}{\partial \myx_{i+1}} \qquad (1\leq i \leq n-1).
\end{equation}
Therefore, for all $i \in \{1,\ldots,N-1\}$ we can compute
\begin{align}
     \frac{\partial C
     }{\partial \myx_i} & = \eta \myx_i + (1-\kappa) D_{i-1} + (1-\kappa) \sum_{n=i+1}^N \frac{\partial D_{n-1}}{\partial \myx_i} \myx_n - 2 \mynu \sum_{n=i}^N X_n \nonumber \\
    & = \eta \myx_i + (1-\kappa) D_{i-1} + (1-\kappa) \left( \frac{\partial D_{i}}{\partial \myx_i} u_{i+1} + \sum_{n=i+2}^N \frac{\partial D_{n-1}}{\partial \myx_i} \myx_n \right) - 2 \mynu \sum_{n=i}^N X_n \nonumber \\
    & = \eta \myx_i + (1-\kappa) D_{i-1} + (1-\kappa) \eta \myx_{i+1} + (1-\kappa)^2 \sum_{n=i+2}^N \frac{\partial D_{n-1}}{\partial \myx_{i+1}} \myx_n - 2 \mynu \sum_{n=i}^N X_n,
    \label{proof:benchmark_det:pC1}
\end{align}
where we used~\eqref{proof:benchmark_det:pD} in the last step. Similarly, 
\begin{equation} \label{proof:benchmark_det:pC2}
\frac{\partial C}{\partial \myx_{i+1}} = \eta \myx_{i+1} + (1-\kappa) D_{i} + (1-\kappa) \sum_{n=i+2}^N \frac{\partial D_{n-1}}{\partial \myx_{i+1}} \myx_n - 2 \mynu \sum_{n=i+1}^N X_n.
\end{equation}
Hence, using~\eqref{proof:benchmark_det:pC2} in~\eqref{proof:benchmark_det:pC1} we obtain the recursive equation
\begin{equation} \label{proof:benchmark_det:pC3}
    \frac{\partial C}{\partial \myx_{i}} = (\eta \myx_i + (1-\kappa) D_{i-1}) (2\kappa+\kappa^2) - 2\mynu\kappa \sum_{n=i+1}^N X_n - 2\mynu X_i + (1-\kappa) \frac{\partial C}{\partial \myx_{i+1}} \quad (1\leq i \leq N-1).
\end{equation}
Next, minimizing~\eqref{proof:benchmark_det:cost} under the constraint $X_0 - \sum_{n=1}^N \myx_n = 0$ and denoting $\lambda \in \mathbb{R}$ the Lagrangian multiplier, we obtain the first order conditions  
\begin{equation}
    \frac{\partial C
     }{\partial \myx_i} = \lambda \qquad (i=1,\ldots,N),
\end{equation}
which, together with~\eqref{proof:benchmark_det:pC3}, can be rewritten as  
\begin{equation} \label{proof:benchmark_det:FOC}
    \lambda \kappa = \kappa (2-\kappa) (\eta \myx_i + (1-\kappa) D_{i-1}) - 2\mynu\kappa \sum_{n=i+1}^N X_n - 2 \mynu X_i \qquad (i=1,\ldots,N-1).
\end{equation}
Computing the differences of the equations in~\eqref{proof:benchmark_det:FOC} for successive $i$ and $i-1$ (for $i \in \{2,\ldots,N-1\}$) and rearranging the terms yields the recursive formula
\begin{align} 
    \myx_i = & \tilde{a} \myx_{i-1} + \tilde{b} D_{i-2} + \tilde{c} X_{i-2} \nonumber \\
    = & \tilde{a} \myx_{i-1} + \tilde{b} (1-\kappa)^{i-2} d_0 + \sum_{j=1}^{i-2} (\tilde{b} \eta (1-\kappa)^{i-2-j} - \tilde{c}) \myx_j + \tilde c X_0 \qquad (i=2,\ldots,N-1), \label{proof:benchmark_det:recursion}
\end{align}
where
\begin{equation} \label{def:constants}
    \tilde{a} := \frac{\kappa (2\kappa\eta-\kappa^2\eta+2\nu)}{2\kappa\eta-\kappa^2\eta-2\kappa\nu+2\nu}, \quad \tilde{b} := \frac{\kappa^2 (2-3\kappa+\kappa^2)}{2\kappa\eta-\kappa^2\eta-2\kappa\nu+2\nu}, \quad \tilde{c} := \frac{-2\kappa\nu}{2\kappa\eta-\kappa^2\eta-2\kappa\nu+2\nu}.
\end{equation}
Due to its linear structure, the recursion in~\eqref{proof:benchmark_det:recursion} can be solved explicitly. We obtain the representation
\begin{equation} \label{proof:benchmark_det:ui}
    \myx_i = \mathfrak{a}_i \myx_{1} + \mathfrak{b}_i d_{0} + \mathfrak{c}_i X_{0} \qquad (i=1,\ldots, N)
\end{equation}
where $\mathfrak{a}, \mathfrak{b}, \mathfrak{c} \in \mathbb{R}^N$ are given by $\mathfrak{a}_1 := 1, \mathfrak{b}_1 := \mathfrak{c}_1 := 0, \mathfrak{a}_2 := \tilde{a}, \mathfrak{b}_2 := \tilde{b}, \mathfrak{c}_2 := \tilde{c}$,
\begin{equation} \label{def:mathcalabc}
\begin{aligned}
\mathfrak{a}_i := & \, \tilde{a} \cdot \mathfrak{a}_{i-1} + \sum_{j=1}^{i-2} \left( \tilde{b} \eta (1-\kappa)^{i-2-j} - \tilde{c} \right) \mathfrak{a}_j, \\
\mathfrak{b}_i := & \, \tilde{a} \cdot \mathfrak{b}_{i-1} + \tilde{b} (1-\kappa)^{i-2} + \sum_{j=1}^{i-2} \left( \tilde{b} \eta (1-\kappa)^{i-2-j} - \tilde{c} \right) \mathfrak{b}_j, \\
\mathfrak{c}_i := & \, \tilde{a} \cdot \mathfrak{c}_{i-1} + \tilde{c} + \sum_{j=1}^{i-2} \left( \tilde{b} \eta (1-\kappa)^{i-2-j} - \tilde{c} \right) \mathfrak{c}_j,
\end{aligned} 
\end{equation}
for $i = 3,\ldots,N-1$, as well as $\mathfrak{a}_N := -\sum_{j=1}^{N-1} \mathfrak{a}_j, \mathfrak{b}_N := -\sum_{j=1}^{N-1} \mathfrak{b}_j, \mathfrak{c}_N := 1-\sum_{j=1}^{N-1} \mathfrak{c}_j$, which comes from the terminal condition $u_N = X_{N-1}$. Moreover,~\eqref{proof:benchmark_det:ui} implies the representation
\begin{equation}
X_i = \mathfrak{a}^x_{i+1} \myx_1 + \mathfrak{b}^x_{i+1} d_0 + \mathfrak{c}^x_{i+1} X_0 \qquad (i=0,\ldots,N-1),
\end{equation}
where $\mathfrak{a}^x, \mathfrak{b}^x, \mathfrak{c}^x \in \mathbb{R}^N$ are given by $\mathfrak{a}^x_1 := \mathfrak{b}^x_1 := 0, \mathfrak{c}^x_1 := 1$ and 
\begin{equation} \label{def:mathcalabcx}
\mathfrak{a}^x_{i+1} := -\sum_{j=1}^{i} \mathfrak{a}_j, \quad \mathfrak{b}^x_{i+1} := -\sum_{j=1}^{i} \mathfrak{b}_j, \quad \mathfrak{c}^x_{i+1} := 1 - \sum_{j=1}^{i} \mathfrak{c}_j
\quad \quad (i = 1,\ldots,N-1);
\end{equation}
as well as
\begin{equation}
D_i = \mathfrak{a}^d_{i+1} \myx_1 + \mathfrak{b}^d_{i+1} d_0 + \mathfrak{c}^d_{i+1} X_0
\qquad (i=0,\ldots,N-1),
\end{equation}
where $\mathfrak{a}^d, \mathfrak{b}^d, \mathfrak{c}^d \in \mathbb{R}^N$ are defined as $\mathfrak{a}^d_1 := 0, \mathfrak{b}^d_1 := 1, \mathfrak{c}^d_1 := 0$ and
\begin{equation}  \label{def:mathcalabcd}
\begin{aligned}
\mathfrak{a}^d_{i+1} := & \, \sum_{j=1}^{i} \eta (1-\kappa)^{i-j} \mathfrak{a}_j , \;
\mathfrak{b}^d_{i+1} := (1-\kappa)^i + \sum_{j=1}^{i} \eta (1-\kappa)^{i-j} \mathfrak{b}_j \\
\mathfrak{c}^d_{i+1} := & \, \sum_{j=1}^{i} \eta (1-\kappa)^{i-j} \mathfrak{c}_j
\end{aligned} 
\quad \quad (i = 1,\ldots,N-1).
\end{equation}
In other words, the first order conditions in~\eqref{proof:benchmark_det:FOC}, together with the terminal state constraint $X_N = 0$, allow to express $u_2,\ldots,u_N$ explicitly in terms of $u_1$ via~\eqref{proof:benchmark_det:ui} and it remains to minimize the costs in~\eqref{proof:benchmark_det:cost} as a linear quadratic function in $u_1$ only, namely   
\begin{align}
C(\myx_1,\ldots,\myx_N) = & \frac{\eta}{2} \sum_{n=1}^N \myx_n^2 + \nu \sum_{n=1}^N (X_{n-1} - \myx_n)^2 + (1-\kappa) \sum_{n=1}^N D_{n-1} \myx_n  \nonumber \\
= & \frac{\eta}{2} \sum_{n=1}^N (\mathfrak{a}_n \myx_{1} + \mathfrak{b}_n d_{0} + \mathfrak{c}_n X_{0})^2 + \nu \sum_{n=1}^{N-1} (\mathfrak{a}^x_{n+1} \myx_1 + \mathfrak{b}^x_{n+1} d_0 + \mathfrak{c}^x_{n+1} X_0)^2 \nonumber \\
& + (1-\kappa) \sum_{n=1}^N (\mathfrak{a}^d_{n} \myx_1 + \mathfrak{b}^d_{n} d_0 + \mathfrak{c}^d_{n} X_0) (\mathfrak{a}_n \myx_{1} + \mathfrak{b}_n d_{0} + \mathfrak{c}_n X_{0}) \nonumber \\
 = & \bigg[ \frac{1}{2} \eta \mathfrak{a}^\top \mathfrak{a} + (1-\kappa) (\mathfrak{a}^d)^\top \mathfrak{a} + \nu (\mathfrak{a}^x)^\top \mathfrak{a}^x) \bigg] u^2_1 \nonumber \\
& + \bigg[ \left(\mathfrak{a}^{\top} (\eta\mathfrak{b}+(1-\kappa)\mathfrak{b}^d) + (1-\kappa) (\mathfrak{a}^d)^\top \mathfrak{b} + 2 \nu (\mathfrak{a}^x)^\top \mathfrak{b}^x \right) d_0 \bigg. \nonumber \\
& \bigg. \hspace{20pt} + \left( \mathfrak{a}^{\top} (\eta\mathfrak{c}+(1-\kappa)\mathfrak{c}^d) + (1-\kappa) (\mathfrak{a}^d)^\top \mathfrak{c} + 2 \nu (\mathfrak{a}^x)^\top \mathfrak{c}^x \right) X_0 \bigg] u_1 \nonumber \\
& + \left( \frac{1}{2}\eta \mathfrak{b}^\top \mathfrak{b} + \nu (\mathfrak{b}^x)^\top \mathfrak{b}^x + (1-\kappa) (\mathfrak{b}^d)^\top \mathfrak{b} \right) d_0^2 \nonumber \\ 
& + \left( \frac{1}{2}\eta \mathfrak{c}^\top \mathfrak{c} + \nu (\mathfrak{c}^x)^\top \mathfrak{c}^x + (1-\kappa) (\mathfrak{c}^d)^\top \mathfrak{c} - \nu \right) X_0^2 \nonumber \\
& + \left( \eta \mathfrak{b}^\top \mathfrak{c} + 2 \nu (\mathfrak{b}^x)^\top \mathfrak{c}^x + (1-\kappa) ((\mathfrak{b}^d)^\top \mathfrak{c} + (\mathfrak{c}^d)^\top \mathfrak{b}) \right) X_0 d_0 \nonumber
\end{align}
which yields the claims in~\eqref{prop:benchmark_det_x1},~\eqref{prop:benchmark_det_xi},~\eqref{prop:benchmark_det_XD}. \qed

\textbf{Proof of Corollary~\ref{cor:benchmark_det}:} In the case $\nu = 0$ the constants introduced in~\eqref{def:constants} reduce to
$\tilde{a} = \kappa$, $\tilde{b}=\kappa(2-3\kappa+\kappa^2)/(2\eta-\kappa\eta)$ and $\tilde{c} = 0$. Moreover, direct computations reveal that the coefficients in~\eqref{def:mathcalabc} simplify to
\begin{equation*}
    \mathfrak{a}_i = \kappa, \quad \mathfrak{b}_i = \tilde{b}, \quad \mathfrak{c}_i = 0 \qquad (i=2,\ldots,N-1),
\end{equation*}
as well as $\mathfrak{a}_N = -1-(N-2)\kappa$, $\mathfrak{b}_N = -(N-2)\tilde{b}$, $\mathfrak{c}_N=1$. This yields the claim in~\eqref{cor:benchmark_det_opt_1}. For the coefficients in~\eqref{def:mathcalabcx} we obtain
\begin{equation*}
    \mathfrak{a}^x_i = -1-(i-1)\kappa, \quad \mathfrak{b}^x_i = -(i-1)\tilde{b}, \quad \mathfrak{c}^x_i = 1 \qquad (i=2,\ldots,N),
\end{equation*}
and for the coefficients in~\eqref{def:mathcalabcd} we have
\begin{equation*}
    \mathfrak{a}^d_i = \eta, \quad \mathfrak{b}^d_i = 1-\kappa, \quad \mathfrak{c}^d_i = 0 \qquad (i=2,\ldots,N).
\end{equation*}
This yields the claim in~\eqref{cor:benchmark_det_xd}. Finally, directly computing $\hat{b}$ and $\hat{c}$ as defined in~\eqref{prop:benchmark_det_bc} gives the expressions in~\eqref{cor:benchmark_det_bc}. 
\qed

\printbibliography

\end{document}